%% file: main.tex
\documentclass[11pt]{article}
\usepackage[margin=1in]{geometry}
\usepackage{microtype}
\usepackage{subfigure}
\usepackage{booktabs} % for professional tables

\usepackage{amsmath}
\usepackage{amssymb}
\usepackage{mathtools}
\usepackage{amsthm}

\usepackage{url}

\usepackage{amsfonts}       % blackboard math symbols
\usepackage{nicefrac}       % compact symbols for 1/2, etc.
\usepackage{microtype}      % microtypography
\usepackage{xcolor}         % colors
\usepackage{enumitem}
\usepackage{xspace}
\usepackage{thm-restate}

\theoremstyle{plain}
\newtheorem{theorem}{Theorem}[section]
\newtheorem{proposition}[theorem]{Proposition}
\newtheorem{lemma}[theorem]{Lemma}

\newtheorem{corollary}[theorem]{Corollary}
\theoremstyle{definition}
\newtheorem{definition}[theorem]{Definition}

\theoremstyle{remark}

\usepackage{verbatim}
\usepackage{graphicx}
\usepackage{algorithm}
\usepackage{algorithmic}

\newcommand{\xnote}[1]{\textcolor{blue}{[Haike: #1]}}

\DeclareMathOperator*{\argmin}{argmin}

\usepackage{hyperref}
\hypersetup{
    colorlinks = true,
    citecolor = blue,
    linkcolor = black,
    urlcolor = black,
}
\usepackage[capitalize,noabbrev]{cleveref}

\title{Graph-Based Algorithms for Diverse Similarity Search}
\author{ 
Piyush Anand\thanks{The authors are sorted alphabetically.} \\ Microsoft\\ piyush.anand@microsoft.com
\and Piotr Indyk \\ MIT \\ indyk@mit.edu
\and Ravishankar Krishnaswamy \\ Microsoft Research \\ rakri@microsoft.com \vspace{0.3cm}
\and Sepideh Mahabadi \\ Microsorft Research \\ smahabadi@microsoft.com
\and Vikas C. Raykar \\ Microsoft \\ vikasraykar@microsoft.com
\and Kirankumar Shiragur \\ Microsoft Research \\ kshiragur@microsoft.com\vspace{0.3cm}
\and Haike Xu \\ MIT \\ haikexu@mit.edu
}
\date{}

\begin{document}
\maketitle
\begin{abstract}
Nearest neighbor search is a fundamental data structure problem with many applications in machine learning, computer vision, recommendation systems and other fields. Although the main objective of the data structure is to quickly report data points that are closest to a given query, it has long been noted~\cite{carbonell1998use}  that without additional constraints the reported answers can be redundant and/or duplicative. This issue is typically addressed in two stages: in the first stage, the algorithm retrieves a (large) number $r$ of points closest to the query, while in the second stage, the $r$ points are post-processed and a small subset is selected to maximize the desired diversity objective. Although popular, this method suffers from a fundamental efficiency bottleneck, as the set of points retrieved in the first stage often needs to be much larger than the final output.

In this paper we present provably efficient algorithms for approximate nearest neighbor search with diversity constraints that bypass this two stage process. Our algorithms are based on popular graph-based methods, which allows us to ``piggy-back'' on the existing efficient implementations.  These are the first graph-based algorithms for nearest neighbor search with diversity constraints.   For data sets with low intrinsic dimension, our data structures report a diverse set of $k$ points approximately closest to the query, in time that only depends on $k$ and $\log \Delta$, where $\Delta$ is the ratio of the diameter to the closest pair distance in the data set. This bound is qualitatively similar to the best known bounds for standard (non-diverse) graph-based algorithms. Our experiments show that the search time of our algorithms is substantially lower than that using the standard two-stage approach. 
\end{abstract}

\section{Introduction}
Nearest neighbor search is a classic data structure problem with many applications in machine learning, computer vision, recommendation systems and other areas~\cite{shakhnarovich2005nearest}. It is defined as follows: given a set $P$ of $n$ points from some space $X$ equipped with a distance function $D(\cdot, \cdot)$,  build a data structure that, given any query point $q \in X$, returns a point $p \in P$ that minimizes $D(q,p)$. In a more general version of the problem we are given a parameter $k$, and the goal is to report $k$ points in $P$ that are closest to $q$. In a typical scenario, the metric space $(X,D)$ is the $d$-dimensional space, and $D(p,q)$ is the Euclidean distance between points $p$ and $q$.
%which measures how dis-similar $p$ and $q$ are from each other. 

Since for high-dimensional point sets the known {\em exact} nearest neighbor search data structures are not efficient,   several approximate versions of this problem have been formulated. A popular theoretical formulation relaxes the requirement that the query algorithm must return the exact closest point $p$, and instead allows it to output any point $p' \in P$ that is a $c$-approximate nearest neighbor of $q$ in $P$, i.e., $D(q,p') \le c D(q,p)$.  In empirical studies,  the quality of the set of points reported by an approximate data structure is measured by its recall, i.e., the average fraction of the true $k$ nearest neighbors returned by the data structure.

Although minimizing the distance of the reported points to the query is often the main objective, it has long been noted~\cite{carbonell1998use} that, without additional constraints, the reported answers are often redundant and/or duplicative. This is particularly important in applications such as recommendation systems or information retrieval, where many similar variants of the same product, product seller, or document exist. For example, an update to the search results listing algorithm implemented by Google in 2019 ensures that ``no more than two pages from the same site'' are listed~\cite{google-div}. Such constraints can be formulated by assuming that each point is assigned a ``color'' (e.g., site id or product seller) and requiring the data structure to output a set $S$ of $k$ points containing {\em at most $k'$ points of each color},  whose distances to $q$ are (approximately) optimal. A more general formulation allows an arbitrary {\em diversity metric $\rho$} (typically different from $D$), and requires the data structure to report a set $S$ of $k$ points such that for any distinct $p,p' \in S$, $\rho(p,p') \ge C$, for some required diversity parameter $C>0$.

The aforementioned paper of~\cite{carbonell1998use}  stimulated the development of the rich area of {\em diversity-based reranking}, which became the dominant approach to this problem. The approach proceeds in two stages. In the first stage, the data structure retrieves $r$ points closest to the query, where $r$ can be much larger than the desired output $k$. In the second stage, the $r$ points are post-processed to maximize the diversity objective of the reported $k$ points.

Though popular, the reranking approach to diversifying nearest neighbor search suffers from a fundamental efficiency bottleneck, as the algorithm needs to retrieve a large enough set to ensure that it contains the $k$ diverse points. In many scenarios, the number $r$ of points that need to be retrieved can be much larger than $k$ (see e.g., \Cref{fig:percentage} and the discussion in the experimental section). In the worst case, it might be necessary to set $r=\Omega(n)$ to ensure that the optimal set is found.
%(PICTURE). 
% No space for picture.
This leads to the following algorithmic question:

\begin{quote} {\em 
    Is it possible to bypass the standard reranking pipeline by directly reporting the $k$ diverse points, in time that depends on $k$ and not $r$?}
\end{quote}

In this paper, we aim to solve this problem both in theory and in practice.  Because of these dual goals,  we focus on designing efficient {\em graph-based} algorithms for diverse similarity search.  In graph-based algorithms,   the data structure consists of a graph between the points in $P$, and the query procedure performs greedy search over this graph to find points close to the query.  Graph-based algorithms such as HNSW~\cite{malkov2018efficient}, NGT~\cite{iwasaki2018optimization}, and DiskANN~\cite{jayaram2019diskann} have become popular tools in practice, often topping Approximate Nearest Neighbor benchmarks~\cite{ANN}. In addition, they are highly versatile, as they do not put any restrictions on the distance function $D$. Although most of the work in this area is purely empirical, a recent paper~\cite{indykxu2024worst} gives approximation and running time guarantees for some of those algorithms. 

\subsection{Our Results} 
We give a positive answer to the aforementioned question, by designing a variant of the DiskANN algorithm that reports approximate nearest neighbors of a given query satisfying diversity constraints. Our theoretical analysis assumes the same setup as in~\cite{indykxu2024worst}. Specifically, we assume that the input point set $P$ has bounded {\em doubling dimension}\footnote{Doubling dimension is a measure of the intrinsic dimensionality of the pointset - see Preliminaries for the formal definition.} $d$, and that its aspect ratio (the ratio of the diameter to the closest pair distance) is at most $\Delta$. Under this assumption, we show that the query time of our data structures is polynomial in $k$, $\log n$ and $\log \Delta$.

Formally, our result is as follows. (Here we state the result in the simplest  setting, where the diversity is induced by point colors and $k'=1$. (See Theorem \ref{thm:diverse_ann} for the general result statement.)

\begin{theorem}\label{thm:colorful_ann}
 Consider the data structure constructed by Algorithm~\ref{alg:color-indexing}. Given a query $q$, let $R$ be the radius of the smallest ball around $q$ (w.r.t. the metric $D$) which contains $k$ points of different colors. Then the search Algorithm~\ref{alg:color-search} returns a set $S$ of $k$ points of different colors such that, for all $p \in S$,
 \[ D(q,p) \le \left(\frac{\alpha+1}{\alpha-1}+\epsilon\right) R. \]
 The search algorithms takes $O\left(k\log_{\alpha}\frac{\Delta}{\epsilon}\right)$ steps, where each step takes $\tilde{O}\left(k(8\alpha)^d\log \Delta\right)$ time.
The data structure uses space $O(n k(8\alpha)^d\log\Delta)$.
\end{theorem}

We note that the approximation factor with respect to $D$, as well as the running time bounds, are essentially the same as the bounds obtained in~\cite{indykxu2024worst} for the 
non-diverse approximate nearest neighbor problem. 
The main difference is that the bound in~\cite{indykxu2024worst} does not depend on $k$, as they consider only the case of $k=1$.

As mentioned earlier, Theorem~\ref{thm:colorful_ann} generalizes to arbitrary $k'$ and general diversity metric $\rho$, as discussed later.

\paragraph{Experimental results.} 
%\snote{Ravi/Kiran/Piyush can you please check this paragraph?}
We adapt our theoretical algorithms to devise fast heuristics based on them, and show the efficiency of our algorithms on several realistic scenarios. In one of them, we consider the task of showing ads to a user based on their search queries. Given a number $k$, say $100$, of available slots, the goal is to choose ads from a large corpus, such that the {\em sellers} (i.e., colors) of those ads are all different. A more relaxed constraint requires that the number of ads shown from a single seller be bounded by at most $k'$, say $10$. To achieve $95\%$ recall$@100$ on this real-world scenario, the prevailing baseline approach of retrieving a much larger number of candidates using a regular ANNS index and then choosing the best diversity-preserving $k$-sized subset of them has latency upwards of 8ms; our algorithm, on the other hand achieves a similar recall at a latency of around 1.5ms, resulting in an improvement upwards of {\bf 5X}! { A production-quality implementation of our algorithm is currently under development for serving large-scale workloads at one of the major technology companies.}
%while the improved search algorithm alone 
%brings it down to approximately 5ms. Making both indexing and search diverse further brings this down to around
%1.5ms, resulting in an improvement upwards of {\bf 5X}.

%The difference between our heuristics and our theoretical algorithms is similar to the difference between the fast- and slow-preprocessing algorithms in DiskANN~\citep{jayaram2019diskann,indykxu2024worst}). As one can see from the plots in \Cref{fig:recall-latency-10,fig:recall-latency-1}, both the new indexing and the search methods play a
%crucial role in improving the overall search quality. For example, to achieve $95\%$ recall$@100$ for the product
%dataset, the baseline reranking approach retrieving $r \gg k$ nearest neighbors followed by post-processing has latency upwards of 8ms, while the improved search algorithm alone 
%brings it down to approximately 5ms. Making both indexing and search diverse further brings this down to around
%1.5ms, resulting in an improvement upwards of {\bf 5X}.

\paragraph{Generalizations.}
On the theoretical side, we extend our results in several directions listed below. These are shown and proved in Section \ref{sec:general-algo}.
\begin{itemize}
    \item \textbf{Relaxing the diversity requirement.} First, in some applications, the requirement that {\em all} reported results have different colors is too strong. (For example, the aforementioned application to search~\cite{google-div} allows for two points having the same color.)  Therefore, we consider a more general constraint, requiring that no color should appear more than $k'$ times. We show how our results can be generalized to any $1\leq k'\leq k$.
    \item \textbf{Generic metric $\rho$.\footnote{In fact our algorithms work for $\rho$ being any {\em pseudo-metric} allowing $\rho(p_1,p_2)=0$ for two different points $p_1,p_2\in P$, but for simplicity we refer to it as metric, throughout the paper.}} Second, we generalize our results to the case where diversity is defined according to a given metric $\rho$ (also defined over $X$ but potentially different from $D$). Here, given a required diversity parameter $C$, the goal is to report a set $S$ of $k$ closest points to the query such that $min_{p_1,p_2\in S} \rho(p_1,p_2)\geq C$.
    We say that such a set $S$ is {\em $C$-diverse}.
    Note that the color version is the special case where $\rho(p_1,p_2)$ is defined to be $0$ if $p_1$ and $p_2$ are of the same color, and $1$ otherwise.
    \item \textbf{Unifying the two generalizations.} In order to unify the above two results, and incorporate the notion of $k'$ into the generic metric $\rho$, we allow for each point in the reported set $S$ to be ``similar" to at most $k'>1$ other points in $S$. More formally, for any $p \in S$ there should be at most $k'-1$ other points $p' \in S$ such that $\rho(p,p') < C$. We say that such a set $S$ is {\em $(k',C)$-diverse}\footnote{We note that the notion of $(k',C)$-diverse set is a new notion of diversity that strictly generalizes the widely used minimum-pairwise-distance notion for diversity.}. We show how our algorithms can be modified to this most general formulation of the problem. 
    \item \textbf{Primal vs Dual formulations.} Finally, instead of asking for the closest $k$ points to the query satisfying a diversity requirement parameterized by $C$, which we refer to as the \emph{primal} variant of the problem, one can ask the {\em dual} question: Given a radius $R$, find a set of $k$ points within radius $R$ of the query, while maximizing the diversity. 
    We show algorithms for both the primal and dual variants of the most general formulation of the problem.
\end{itemize}

\subsection{Related Work}

Nearest Neighbor Search with diversity requirement has been previously studied in the work of~\cite{abbar2013real,abbar2013diverse}, where they presented a ``diversified version'' of the {\em Locality-Sensitive Hashing (LSH)} approach due to~\cite{indyk1998approximate}. However, their diversification approach does not carry over to the graph-based methods.
Moreover, they provide a bi-criteria approximation only for the {\em dual formulation} of the problem, and do not consider the primal formulation.
Finally, the distance functions $D$ that they consider are limited to Hamming distance or its variants like the Jaccard similarity~\cite{abbar2013real}. Although it is plausible that the result could be extended to other distances that are supported by LSH functions, not all distance functions satisfy this constraint.

\section{Preliminaries}
\paragraph{Problem definition.}

Let $(X,D)$ be the underlying metric space, with distance function $D$. 
Let $P$ be our colored point set. For $p\in P$, we use $col[p]$ to denote its {\em color}. 

\begin{definition}[colorful]\label{def:colorful}
A set $S$ is {\em colorful} if no two points in $S$ have the same color.
\end{definition}

\begin{definition}[$k'$-colorful]\label{def:k'-colorful}
A set $S$ is {\em $k'$-colorful}, if within the multi-set $\{col[p_1],...,col[p_k]\}$, no color appears more than $k'$ times. 
\end{definition}

Note that for $k'$-colorful for $k'=1$ is equivalent to the colorful notion.

\begin{definition}\label{def:S_i}
    Given a subset $S\subset P$ of size $k$, and a query $q$, for each $i\leq k$, we use $S_i(q)$ to denote the distance of the $i$th closest point in $S$ to $q$. When $q$ is clear from the context, we drop $q$ and simply use $S_i$.
%Moreover, for two solutions $\mathsf{ALG}$ and $\mathsf{OPT}$ of size $k$, we write $\mathsf{ALG}\le \mathsf{OPT}$ if for any $i$, $\mathsf{ALG_i}\le \mathsf{OPT_i}$. \xnote{I am not sure whether this notation is useful any more? if not, shall we delete it}
\end{definition}

\begin{definition}[Colorful NN]
    Given a colored point set $P$, the goal of {\em colorful NN} is to preprocess $P$ and create a data structure such that given a query point $q$, one can quickly return a colorful subset $S\subset P$ of size $k$ such that $S_k$ is minimized. 
\end{definition}

Note that, when $k=1$, our Colorful NN degenerates to the standard nearest neighbor search problem.

\begin{definition}[$k'$-Colorful NN] \label{def:kprime-colorful}
    Given a colored point set $P$, the goal of {\em $k'$-colorful NN} is to preprocess $P$ and create a data structure such that given a query point $q$, one can quickly return a $k'$-colorful subset $S\subset P$ of size $k$ such that $S_k$ is minimized. 
\end{definition}

\paragraph{Balls, doubling dimension, and aspect ratio.}
%The distance between any two points $p,q$ is $D(p,q)$. 
We use $B_D(p,r)$ to denote a ball centered at $p$ with radius $r$, i.e.,  $B_D(p,r)=\{u\in X: D(u,p)<r\}$. We will drop the subscript $D$ if the metric is clear from the context. 

We say a point set $P$ has \emph{doubling dimension} $d$ if for any point $p$ and radius $r$, the set $B(p,2r)\cap P$ can be covered by at most $2^d$ balls of radius $r$.

\begin{lemma}\label{lm:doubling_dimension}
Consider any point set $P \subset X$ with doubling dimension $d$.

For any ball $B(p,r)$ centered at some point $p\in P$ with radius $r$ and a constant $\alpha$, we can cover $B(p,r)\cap P$ using at most $m\le O(\alpha^d)$ balls with radius smaller than $r/\alpha$, i.e. $B(p,r)\cap P \subset \bigcup_{i=1}^m B(p_i,r/\alpha)$  for some $p_1 \ldots p_m \in X$. 
%{\color{red} \snote{so here $p$ and $p_i$ do not have anything to do with $P$, right?}} 
%\xnote{I think the large ball is centered at some $p\in P$, but not necessary the smaller balls? Is this correct? @piotr}
% PIOTR: yes, the small balls can be centered anywhere. N
\end{lemma}

We define $\Delta=\frac{D_{max}}{D_{min}}$ to be the \emph{aspect ratio} of the point set $P$ where $D_{max}$($D_{min}$, resp.) represents the maximal (minimal, resp.) distance between any two points in the point set $P$.

%%%%%%%%%%%%%%%%%%%%%%%%%%%%%%%%%%%%%%%%%%%%%
%.................Simple Case

\section{Algorithm for Colorful NN}\label{sec:main-algo}
In this section, for the sake of simplicity of presentation, 
%and its direct connection to the experimental section, 
% PIOTR: experimantal section also uses k'>1
we focus on the simplest setting where $k'=1$, and $\rho$ is the binary metric. The binary setting corresponds to our main application of seller diversity, and the case of $k'=1$ focuses on retrieving $k$ closest points from the data set such that \emph{all} of them have different colors/sellers. 
The algorithm that handles the general setting is presented in Section \ref{sec:general-algo}. 

\paragraph{Intuition behind the algorithm.} 
At a high level, the ``slow pre-processing'' algorithm of~\cite{indykxu2024worst} uses the following intuition when pruning the graph: If $u$ and $v$ are ``much closer'' to each other than to another point $p$, then it is not necessary to connect $p$ to both $u$ and $v$. 
This makes it possible to track the progress of the search procedure as it identifies points closer to the query point and use the doubling dimension to bound the degree of the graph and the total space. 
Our algorithm retains these insights, but also requires several new ones, as now we need to show that the search algorithms can progress {\em while maintaining a colorful solution}. On a high-level, this is established using the following two intuitions. 
First, if the colors of $u$ and $v$ are the same, then again there is no need to connect $p$ to both of them, and we can use the same pruning as before. 
Second, if $p$ is already connected to $k$ points $v_1,\cdots, v_k$, all of which are much closer to $u$ compared to $p$, and that they all have different colors, then again there is no need to connect $p$ to $u$. This is roughly because no solution would need more than $k$ points of different colors in a relatively small neighborhood. 
%This is also why the space usage of our algorithm is higher by a factor $k$ than the algorithm of \cite{indykxu2024worst}.
The main challenge in converting these intutions into a formal argument is in showing that such a graph keeps enough edges for a greedy search algorithm to converge to an approximately optimal solution for the coloful NN problem.

\subsection{The Preprocessing Algorithm} 
The indexing algorithm is shown in Algorithm~\ref{alg:color-indexing}. 
\begin{algorithm}
\caption{Indexing algorithm for colorful NN}
\label{alg:color-indexing}
\begin{algorithmic}[1]
\STATE \textbf{Input}: A set of $n$ points $P=\{p_1,...,p_n\}$; $k$ is the size of the output; $\alpha$ is the parameter used for pruning.
\STATE \textbf{Output}: A directed graph $G=(V,E)$ where $V=\{1,...,n\}$ are associated with the point set $P$.
\FOR{each point $p\in P$}
    \STATE Sort all points $u\in P$ based on their distance from $p$ and put them in a list $\mathcal{L}$ in that order.
    \WHILE{$\mathcal{L}$ is not empty}
        \STATE $u\gets \argmin\limits_{u\in \mathcal{L}}D(u,p)$
        \STATE Initialize $\mathsf{rep}[u]\gets \{u\}$
        \FOR{each point $v\in \mathcal{L}$ in order}
            \IF{$D(u,v)\le D(p,u)/(2\alpha)$}
                \IF{$col[v]$ not shown in $rep[u]$ and $|rep[u]|<k$}
                    \STATE $\mathsf{rep}[u]\gets \mathsf{rep}[u]\cup v$ 
                \ENDIF
                \STATE remove $v$ from $\mathcal{L}$
            \ENDIF
        \ENDFOR
        % \State $\mathsf{rep}[u]\gets$ use the greedy algorithm of Gonzales to choose $k/k'$ points in $\mathsf{bag}[u]$ to approximately maximize the minimum pairwise diversity distance.
        \STATE add edges from $p$ to $\mathsf{rep}[u]$
        \STATE Remove $u$ from $\mathcal{L}$
    \ENDWHILE
\ENDFOR
\end{algorithmic}
\end{algorithm}

\begin{lemma}\label{lm:degree_color}
The graph constructed by Algorithm~\ref{alg:color-indexing} has degree limit $O(k(8\alpha)^d\log\Delta)$. 
\end{lemma}
\begin{proof}
Let's first bound the number of points not removed by others, then according to Line 10 in Algorithm~\ref{alg:color-indexing}, the degree bound will be that times $k$.

We use $\mathsf{Ring}(p,r_1,r_2)$ to denote the points whose distance from $p$ is larger than $r_1$ but smaller than $r_2$. For each $i\in [\log_2 \Delta]$, we consider the $\mathsf{Ring}(p,D_{max}/2^i,D_{max}/2^{i-1})$ separately. According to Lemma~\ref{lm:doubling_dimension}, we can cover $\mathsf{Ring}(p,D_{max}/2^i,D_{max}/2^{i-1})\cap P$ using at most $m\le O((8\alpha)^d)$ small balls of radius $D(p,u)/(4\alpha)\geq \frac{D_{max}}{2^{i+2}\alpha}$. According to the pruning criteria in Line 9, within each small ball, there will be at most one such point $u$ remaining. This establishes the degree bound of $O(k(8\alpha)^d\log\Delta)$.
\end{proof}

\subsection{The Search Algorithm} Algorithm~\ref{alg:color-search} shows the search algorithm for the colorful nearest neighbor search problem. 
Next, we analyze the search algorithm and finally prove Theorem~\ref{thm:colorful_ann}.

\begin{algorithm}[ht]
\caption{Search algorithm for colorful NN}
\label{alg:color-search}
\begin{algorithmic}[1]
\STATE \textbf{Input}: A graph $G=(V,E)$ with $N_{out}(p)$ being the out edges of $p$; query $q$; number of optimization steps $T$.
\STATE \textbf{Output}: A set of $k$ points $\mathsf{ALG}$.
\STATE Initialize $\mathsf{ALG}=\{p_1,...,p_k\}$ to be any set of $k$ points with different colors.
\FOR{$i=1$ to $T$}
    % \State $U\gets \{u | (p_j,u)\in E\ s.t. \exists\ p_j\in ALG\}$
    \STATE $p_k\gets$ the furthest point in $\mathsf{ALG}$ from $q$.
    \STATE $U\gets N_{out}(p_k)$ and sort $U$ based on their distance from $q$
    % \STATE $U\gets \bigcup\limits_{p\in \mathsf{ALG}}(N_{out}(p)\cup p)$ and sort $U$ based on their distance from $q$
    % \STATE $\mathsf{ALG}\gets$ the closest $k-1$ points in $\mathsf{ALG}$    
    \FOR{each point $u\in U$}
        \IF{$\mathsf{ALG}\setminus p_k \bigcup u$ is colorful}
            \STATE $\mathsf{ALG}\gets \mathsf{ALG}\setminus p_k \bigcup u$
            \STATE Break
        \ENDIF
        % \IF{$|\mathsf{ALG}|=k$}
            % \STATE Break
        % \ENDIF
    \ENDFOR
    % \If{$\mathsf{ALG}^{t}\ge \mathsf{ALG}^{t-1}$}
    % \State Break
    % \EndIf
\ENDFOR
\STATE \textbf{Return} $\mathsf{ALG}$
\end{algorithmic}
\end{algorithm}

\begin{proposition}\label{prop:p_star_exists_diverse}
Let $\mathsf{OPT}=\{p^*_1,...,p^*_k\}$ be a colorful solution with minimized $\mathsf{OPT_k}$, and $\mathsf{ALG}=\{p_1,...,p_k\}$ be the current solution (both ordered by distance from $q$). If $p_k\notin \mathsf{OPT}$, there exists a point $p^*\in \mathsf{OPT}\setminus \mathsf{ALG}$ such that $\mathsf{ALG}\setminus p_k\bigcup p^*$ is colorful. 
\end{proposition}
\begin{proof}
Observe that throughout the search algorithm, we maintain the property that $\mathsf{ALG}$ is colorful. 
Note that $\mathsf{ALG}\setminus p_k$ has $k-1$ different colors, and $\mathsf{OPT}$ has $k$ different colors. Thus there should be a point $p^*\in OPT$ whose color is different from all points in $\mathsf{ALG}\setminus p_k$. Note that such $p^*$ cannot belong to $\mathsf{ALG}$ and thus belongs to $\mathsf{OPT}\setminus\mathsf{ALG}$.
%Let $\overline{\mathsf{OPT}}=\mathsf{OPT}\setminus \mathsf{ALG}$. We partition $\overline{\mathsf{OPT}}$ into different color sets $\mathsf{OPT}\setminus \mathsf{ALG}=z_1\cup z_2...\cup z_m$. Simultaneously, we arrange $\mathsf{ALG}\setminus p_k \setminus \mathsf{OPT}$ into the same set of colors $\mathsf{ALG}\setminus p_k \setminus \mathsf{OPT}=z'_1\cup z'_2...\cup z'_m$. Because $|ALG|=|OPT|=k$ and $p_k\notin OPT$, we have $|\mathsf{ALG}\setminus p_k \setminus \mathsf{OPT}|<|\mathsf{OPT}\setminus\mathsf{ALG}|$, so there exists at least a color $i$ such that $|z'_i|<|z_i|$. Therefore, we can pick at least one point from $p^*\in z_i\setminus z'_i$, so that $\mathsf{ALG}\setminus p_k\bigcup p^*$ is $k$-colorful.
\end{proof}

\begin{lemma}\label{lm:update_color}
There always exists a point $p'\in N_{out}(p_k)$ (for $p_k$ as defined in Line 5) such that 
%(i) $\mathsf{ALG}\setminus p_k \bigcup p'$ is colorful; and (ii) $D(p',q)\le D(p_k,q)/\alpha+\mathsf{OPT_k}(1+1/\alpha)$
%\iffalse
\begin{enumerate}[topsep=0cm]
\setlength\itemsep{0cm}
\item $\mathsf{ALG}\setminus p_k \bigcup p'$ is colorful
\item $D(p',q)\le D(p_k,q)/\alpha+\mathsf{OPT_k}(1+1/\alpha)$
\end{enumerate}
%\fi
\end{lemma}
\begin{proof}
According to Proposition~\ref{prop:p_star_exists_diverse}, for any current solution $\mathsf{ALG}$ with $p_k\notin \mathsf{OPT}$, there exists a point $p^*\in \mathsf{OPT}\setminus \mathsf{ALG}$ such that $\mathsf{ALG}\setminus p_k\cup p^*$ is colorful. 
%Let $w\in \mathsf{ALG}$ be the closest point to $p^*$. 
If there exists an edge from $p_k$
%$w$ 
to $p^*$, replacing $p_k$ with $p^*$ is a potential update. We set $p'=p^*$ and $D(p',q)\le \mathsf{OPT_k}$ satisfies the distance upper bound above. 

Otherwise, we let $u$ be the point where $p^*$ was removed when processing $u$ on line 9 in Algorithm~\ref{alg:color-indexing}. Because $p^*$ was not connected from $p_k$, either there exists a point in $\mathsf{rep}[u]$ with the same color, or $\mathsf{rep}[u]$ has already got $k$ points with different colors. In the first case, we can set $p'$ to be the point in $\mathsf{rep}[u]$ with the same color. In the latter case, by pigeon hole principle, there always exists a color in $\mathsf{rep}[u]$ not shown in $\mathsf{ALG}\setminus p_k$. Therefore, we can find a desired $p'\in \mathsf{rep}[u]$ and it is connected to $p_k$.

We have proved that the $p'$ we found satisfies the colorful criteria. Now we will bound its distance upper bound.
\begin{align}
D(p',q)
&\le D(p^*,q)+D(p',p^*) \notag \\
&\le D(p^*,q)+D(p',u)+D(p^*,u) \notag\\
&\le D(p^*,q)+D(p_k,u)/(2\alpha)+D(p_k,u)/(2\alpha) \tag{Line 9 in Algorithm~\ref{alg:color-indexing}}\\
&\le D(p^*,q)+D(p_k,u)/\alpha \notag \\
&\le D(p^*,q)+D(p_k,p^*)/\alpha \tag{Because $u$ is ordered earlier than $p^*$ in Algorithm~\ref{alg:color-indexing}}\\
&\le D(p^*,q)+D(p_k,q)/\alpha+D(p^*,q)/\alpha \notag\\
&\le D(p_k,q)/\alpha+\mathsf{OPT_k}(1+1/\alpha)\notag
\end{align}
\end{proof}
\begin{proof}[Proof of Theorem~\ref{thm:colorful_ann}]
Regarding the running time, the total number of edges connected from any point in $\mathsf{ALG}$ is bounded by $|U|\le O(k(8\alpha)^d\log \Delta)$. In each step, the algorithm first sorts all these edges connected from $p_k\in\mathsf{ALG}$ and then checks whether each of them can be added to the new $\mathsf{ALG}$ set. The total time spent per step is bounded by $O(|U|\log|U|)$. The overall time complexity is $\tilde{O}\left(k(8\alpha)^d\log\Delta\right)$ per step.

To analyze the approximation ratio, at time step $t$, we use $\mathsf{ALG^t}=\{p^t_1,...,p^t_k\}$ to denote the current unordered solution. We denote $\mathsf{ALG^t_k}=\max\limits_{i\in[k]}D(p^t_i,q)$. According to Algorithm~\ref{alg:color-search} and Lemma~\ref{lm:update_color}, if $p_i$ is updated at time step $t$, we have $D(p^t_i,q)\le D(p^{t-1}_i,q)/\alpha+\mathsf{OPT_k}(1+1/\alpha)$. By an induction argument, if a point $p_i$ is updated by $t$ times at the end of time step $T$, we have $D(p^T_i,q)\le \frac{D(p^0_i,q)}{\alpha^t}+\frac{\alpha+1}{\alpha-1}\mathsf{OPT_k}$. 

We now prove that $\mathsf{ALG^T_k}\le \max\limits_{i}\frac{D(p^0_i,q)}{\alpha^{T/k}}+\frac{\alpha+1}{\alpha-1}\mathsf{OPT_k}$. Let $i\in[k]$ be the index achieving the maximal distance upper bound. For the sake of contradiction, if $\mathsf{ALG^T_k}>\frac{D(p^0_i,q)}{\alpha^{T/k}}+\frac{\alpha+1}{\alpha-1}\mathsf{OPT_k}$, this means that $p^T_i$ was updated for at most $T/k-1$ times. By a counting argument, there exists another index $j$ which was updated for at least $T/k+1$ times. However, at the time $t$ when $p^t_j$ was already updated for $T/k$ times, $D(p^t_j,q)\le \frac{D(p^0_j,q)}{\alpha^{T/k}}+\frac{\alpha+1}{\alpha-1}\mathsf{OPT_k} < \mathsf{ALG^T_k}\le \mathsf{ALG^t_k}$, so the algorithm wouldn't have chosen $p^t_j$ to optimize because it couldn't have had the maximal distance at that time, leading to a contradiction. Therefore, we prove that $\mathsf{ALG^T_k}\le \max\limits_{i}\frac{D(p^0_i,q)}{\alpha^{T/k}}+\frac{\alpha+1}{\alpha-1}\mathsf{OPT_k}$.

Now we consider the following three cases depending on the value of the maximal $D(p^0_i,q)$. This case analysis is similar to the proof in Theorem 3.4 from~\cite{indykxu2024worst}.
\begin{enumerate}

\item[Case 1:]  $D(p^0_i,q)>2D_{max}$. 

Let $p^*_k$ be the point having the maximal distance from $q$ in an optimal solution $\mathsf{OPT}$. We know that for any $p^0_i$, we have $D(p^*_k,q)\ge D(p^0_i,q)-D(p^0_i,p^*_k)\ge D(p^0_i,q)-D_{max}\ge D(p^0_i,q)/2$. Therefore, the approximation ratio after $T$ optimization steps is upper bounded by $\frac{\mathsf{ALG^T_k}}{D(p^*_k,q)}\le \frac{D(p^0_i,q)}{D(p^*_k,q)\alpha^{T/k}}+\frac{\alpha+1}{\alpha-1}\le \frac{2}{\alpha^{T/k}}+\frac{\alpha+1}{\alpha-1}$. A simple calculation shows that we can get a $(\frac{\alpha+1}{\alpha-1}+\epsilon)$ approximate solution in $O(k\log_{\alpha}\frac{2}{\epsilon})$ steps.

\item[Case 2:] $D(p^0_i,q)\le 2D_{max}$ and $\mathsf{OPT_k}>\frac{\alpha-1}{4(\alpha+1)}D_{min}$. 

To satisfy $\frac{D(p^0_i,q)}{\alpha^{T/k}}+\frac{\alpha+1}{\alpha-1}\mathsf{OPT_k}\le (\frac{\alpha+1}{\alpha-1}+\epsilon)\mathsf{OPT_k}$, we need $\frac{D(p^0_i,q)}{\alpha^{T/k}}\le \epsilon \mathsf{OPT_k}$. Applying the lower bound $\mathsf{OPT_k}\ge \frac{\alpha-1}{4(\alpha+1)}D_{min}$, we can get that $T\ge k\log_{\alpha}\frac{2(\alpha+1)\Delta}{(\alpha-1)\epsilon}$ suffices.

\item[Case 3:] $D(p^0_i,q)\le 2D_{max}$ and $\mathsf{OPT_k}\le \frac{\alpha-1}{4(\alpha+1)}D_{min}$. 

In this case, we must have $k=1$, because otherwise $D(p^*_k,p^*_1)\le 2D(p^*_k,q)<D_{min}$, violating the definition of $D_{min}$. Suppose $k=1$ and the problem degenerates to the standard nearest neighbor search problem. After $T$ optimization steps, if $p^T_1$ is still not the exact nearest neighbor, we have $D(p^T_1,q)\ge D(p^T_1,p^*_1)-\mathsf{OPT_1}\ge \frac{D_{min}}{2}$. Applying the upper bound of $D(p^T_1,q)$ and $\mathsf{OPT_1}$, we have $\frac{D_{min}}{2}\le D(p^T_1,q)\le \frac{D(p^0_1,q)}{\alpha^{T}}+\frac{\alpha+1}{\alpha-1}\mathsf{OPT_1}\le \frac{D(p^0_1,q)}{\alpha^{T}}+\frac{D_{min}}{4}$. This can happen only if $T\le \log_{\alpha}\frac{\Delta}{8}$.

In conclusion, $O(k\log_{\alpha}\frac{\Delta}{\epsilon})$ steps suffice to achieve the desired approximation ratio in Theorem~\ref{thm:diverse_ann}.
\end{enumerate}
\end{proof}

\subsection{High-level Intuition about the Generalizations}
%\snote{Piotr, please check here}
%PIOTR: done. looks good!
Given our results on colorful NN, it is relatively simple to extend them to the $k'$-colorful NN version with the same bounds. \emph{One key contribution is to demonstrate that, in the graph degree bound, the overhead factor $k$ can be reduced to  $k/k'$ while preserving the approximation quality}. This reduces both the query time bound and the overall space used by the algorithm. This improvement is tight, in the following sense: When $k' = 1$, we recover the bound for colorful NN problem, and when $k' = k$, we recover the standard $k$-NN bound, where no additional factor is needed.

For our algorithm to work with a generic diversity metric $\rho$, we use an intuition similar to that in colorful case. However, instead of pruning an edge from the point $p$ to the point $u$ when $p$ is already connected to representative vectors $v_1,\cdots,v_k$ of different colors, we now choose the representatives based on the diversity metric $\rho$. We find a diverse subset of points in the neighborhood of $u$ (e.g., using the greedy Gonzales algorithm for the k-center problem) and connect $p$ only to those selected points $v_1,\cdots,v_k$.
Again, the main challenge is to demonstrate that a greedy search algorithm can converge to an approximately optimal solution, given the set of edges we retain. The difficulty lies in the fact that we can only maintain an {\em approximately} diverse subset $\mathsf{ALG}$, in contrast to the colorful version, where we only needed $\mathsf{ALG}$ to contain $k$ different colors.  As the algorithm proceeds with further iterations, the technical difficulty lies in ensuring that the approximation factor does not grow depending on the number of iterations.
%we need to ensure that the approximation factor does not grow, and instead can be bounded by a constant not depending on the number of iterations.
%%%%%%%%%%%%%%%%%%%%%%%%%%%%%%%%%%%%%%%%%%%%5
%%%%%%%%%%%%%%%%%%%%%%%%%%%%%%%%%%%%%%%%%%%%
\input{expirements}

\onecolumn

\bibliography{main}
\bibliographystyle{alpha}

%%%%%%%%%%%%%%%%%%%%%%%%%%%%%%%%%%%%%%%%%%%%%%%%%%%%%%%%%%%%%%%%%%%%%%%%%%%%%%%
%%%%%%%%%%%%%%%%%%%%%%%%%%%%%%%%%%%%%%%%%%%%%%%%%%%%%%%%%%%%%%%%%%%%%%%%%%%%%%%
% APPENDIX
%%%%%%%%%%%%%%%%%%%%%%%%%%%%%%%%%%%%%%%%%%%%%%%%%%%%%%%%%%%%%%%%%%%%%%%%%%%%%%%
%%%%%%%%%%%%%%%%%%%%%%%%%%%%%%%%%%%%%%%%%%%%%%%%%%%%%%%%%%%%%%%%%%%%%%%%%%%%%%%

\appendix

\input{generalized-alg}

\input{heuristic}

\end{document}

%% file: expirements.tex
\newcommand{\diversesearch}{\ensuremath{\mathrm{DiverseSearch}}}
\newcommand{\diverseindex}{\ensuremath{\mathrm{DiverseIndex}}}
\newcommand{\diverseprune}{\ensuremath{\mathrm{DiversePrune}}}
\newcommand{\diversequeue}{\ensuremath{\mathrm{DiversePriorityQueue}}\xspace}
\newcommand{\blockers}{\mathsf{blockers}\xspace}

\section{Experimental Evaluation}
In this section we provide an empirical evaluation of our methods. To this end, we first devise a heuristic adaptation based on our provable algorithms for the {\em $k'$-colorful NN} problem as in~\Cref{def:kprime-colorful}. As we note below, this problem itself captures several real-world notions of diversity. 

At a high level, the difference between our heuristics and our theoretical algorithms is similar to the difference between the fast- and slow-preprocessing algorithms in DiskANN~\cite{jayaram2019diskann,indykxu2024worst}). Indeed, we deploy the same construction as in the fast-preprocessing variant of DiskANN, but modify the pruning procedure to insist that any node $u$ has sufficiently many colorful out-neighbors before an edge $(u,v)$ can get pruned, in addition to the geometric condition for pruning as in the original algorithm~\cite{DiskANN}. The number of colorful edges that  are needed before pruning can occur is given by a tunable parameter $m$ in our algorithm, and indeed this is the direct heuristic analog of step 10 in~\Cref{alg:color-indexing}.  This is a very high-level description, and we refer the interested reader to~\Cref{sec:impl} for the complete pseudo-code of our heuristic algorithms.

%Crucially, since our algorithm only adds edges on top of vanilla DiskANN, the same index can be successfully (by restricting to the original set of edges) used for scenarios where diversity is \emph{not required at search time}, significantly saving cost and complexity in real-world deployments. 

Second, we run experiments using our heuristics	on several different datasets, both real-world as well as synthetically generated. We show how our heuristic consistently delivers superior recall for a fixed latency budget, across datasets when compared to a natural baseline of using a vanilla DiskANN algorithm and enforcing diversity only via a final post-processing. We stress that both our real-world datasets are motivated from important shopping scenarios: the data points represent products and a color of a vector corresponds to either the seller or the brand of the product. 
It is then desirable to output results from a diverse set of sellers/brands~\cite{google-div}. Intuitively, displaying diverse results would lead to increased competition between the sellers, and also simultaneously higher click probabilities, thereby leading to an increase in revenue of the exchange.%

\subsection{Experiment Setup} 
All experiments were run on a Linux Machine with AMD Ryzen Threadripper 3960X 24-Core Processor CPU's @ 2.3GHz with 48 vCPUs and 250 GB RAM. All query throughput and latency measurements are reported for runs with 48 threads.

\begin{figure*}[t]
     \centering
    \begin{minipage}{0.40\textwidth}
%        \centering
    \includegraphics[width=\textwidth]{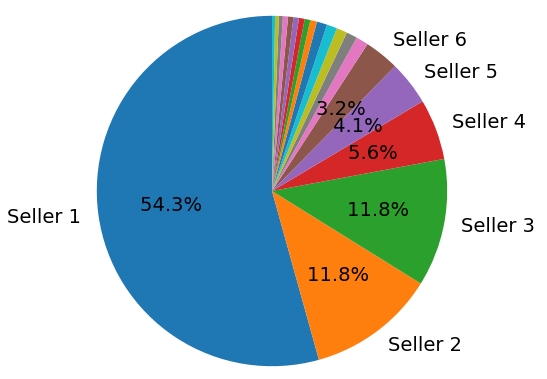}
\caption{Seller distribution in a real-world dataset with 20 million base vectors, where the top 7 sellers constitute more than 90\% of the data.}
\label{fig:percentage}
    \end{minipage}
    \hspace{1cm}
    \begin{minipage}{0.50\textwidth}
%        \centering
    \includegraphics[width=\textwidth]{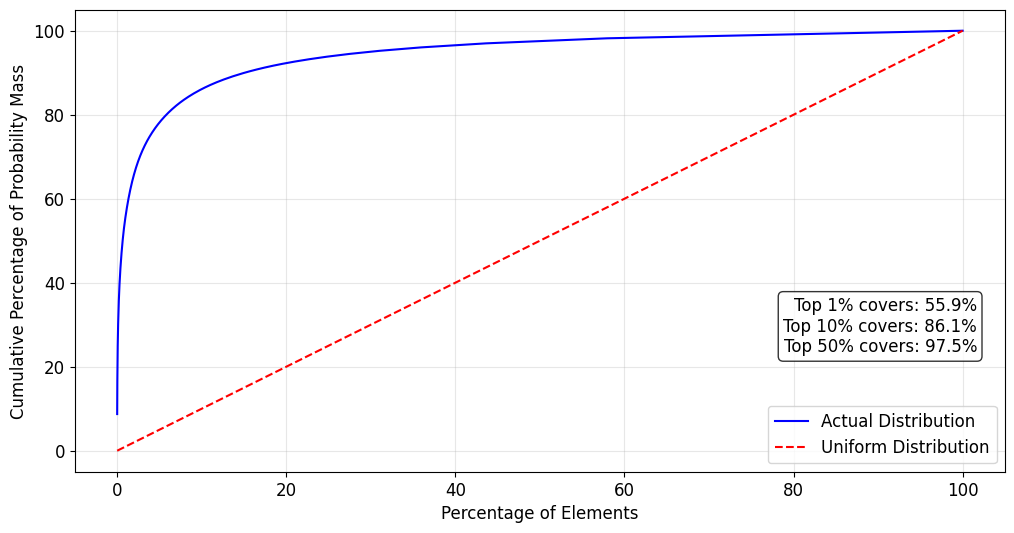}
\caption{Brand cumulative distribution for Amazon dataset, showing the coverage of the vectors by the brands in sorted order. The top $10\%$ of brands cover $86\%$ of the vectors.}
\label{fig:percentage2}
    \end{minipage}
\end{figure*}

\begin{figure*}[!ht]
     \centering
    \begin{minipage}{0.42\textwidth}
        \centering
\includegraphics[width=0.8\textwidth]{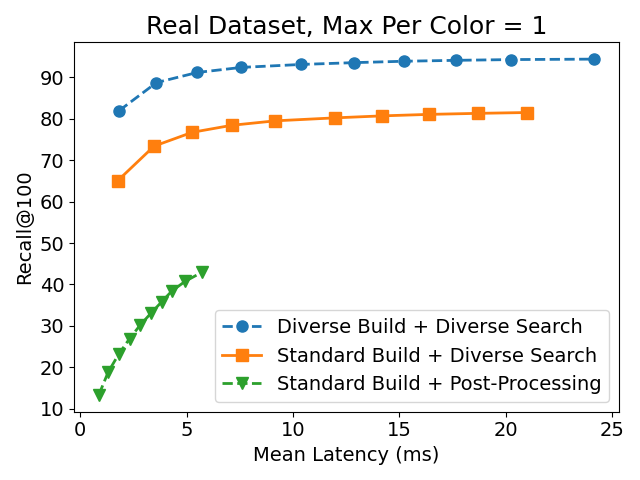}

\includegraphics[width=0.8\textwidth]{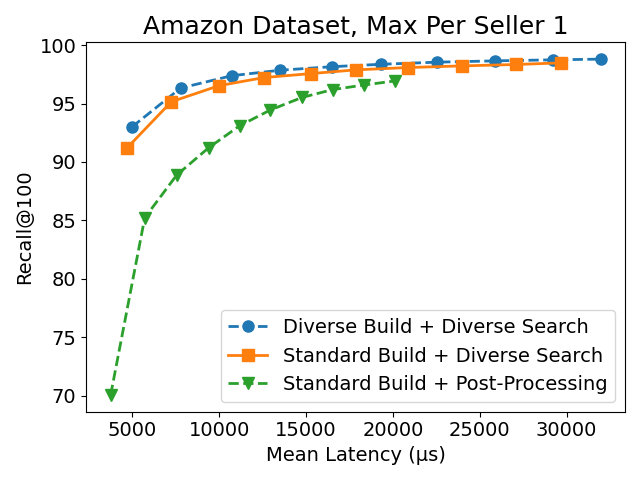}

\includegraphics[width=0.8\textwidth]{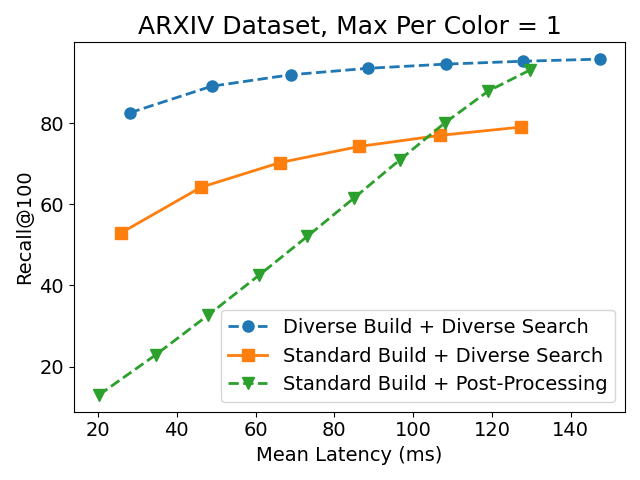}

\includegraphics[width=0.8\textwidth]{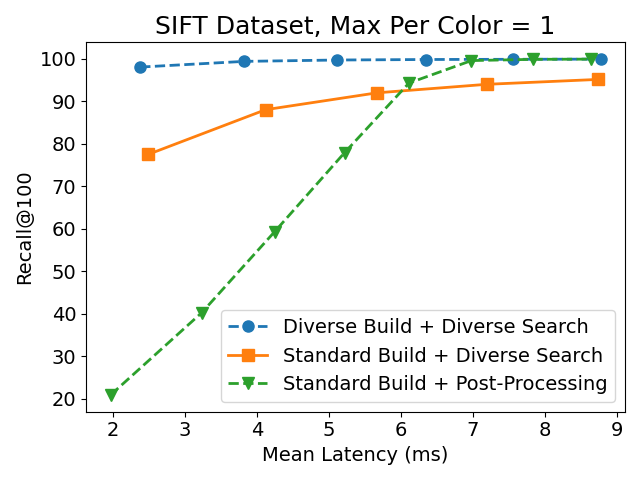}

\includegraphics[width=0.8\textwidth]{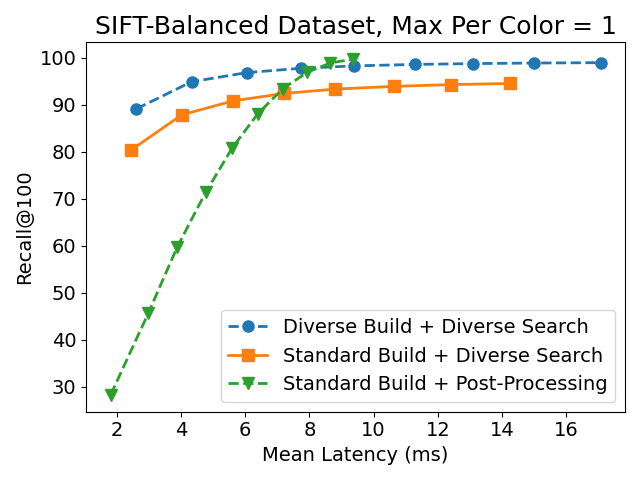}
    \caption{Recall vs Latency for $k'=1$.}
 \label{fig:recall-latency-1}
    \end{minipage}   
%    \hspace{0cm}
    \begin{minipage}{0.42\textwidth}
        \centering
\includegraphics[width=0.8\textwidth]{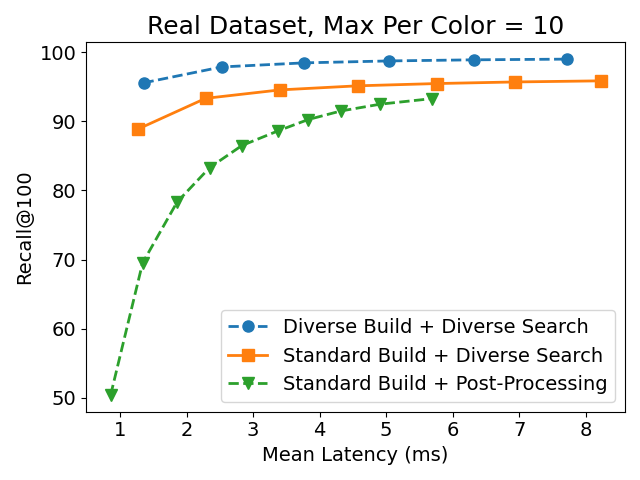}

\includegraphics[width=0.8\textwidth]{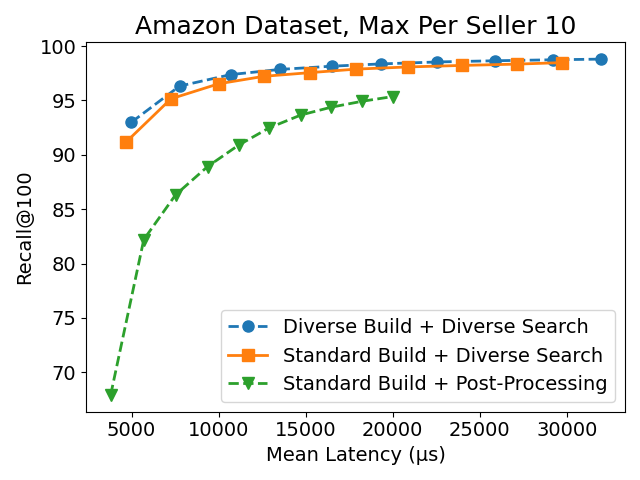}

\includegraphics[width=0.8\textwidth]{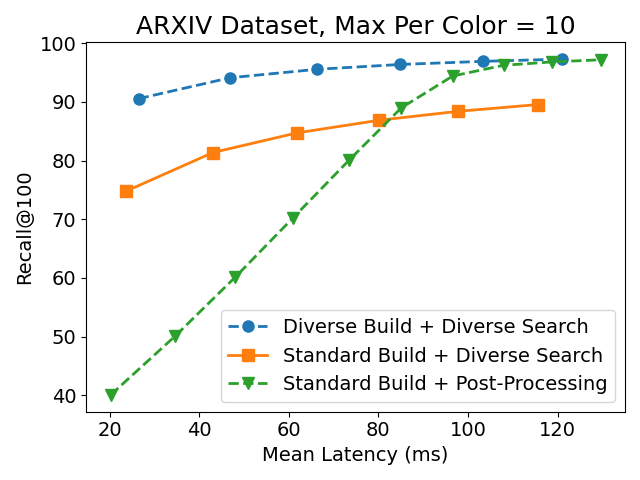}

\includegraphics[width=0.8\textwidth]{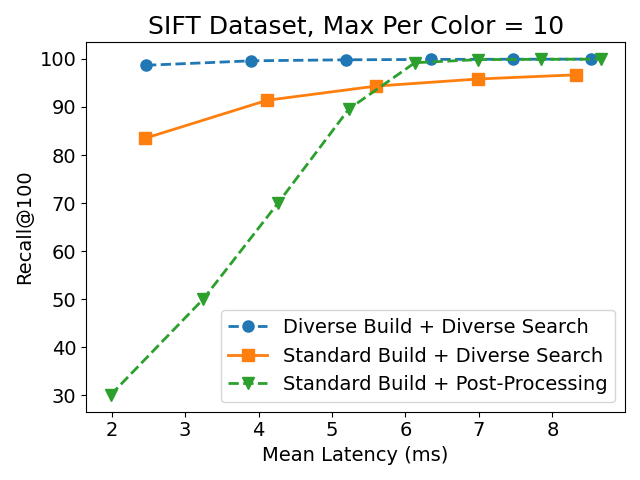}

\includegraphics[width=0.8\textwidth]{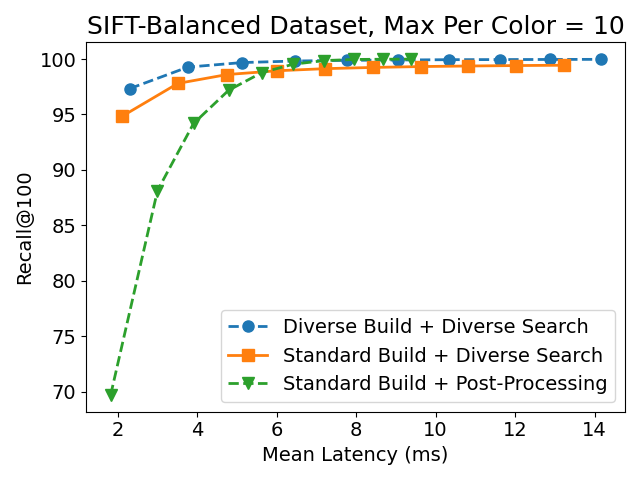}
    
\caption{Recall vs Latency for $k'=10$.}
\label{fig:recall-latency-10}
\end{minipage}
\end{figure*}

\paragraph{Datasets.} We consider two real-world and three semi-synthetic datasets for evaluation.
\begin{itemize}[leftmargin=*]
    \item {\em Real-world Seller Dataset:} Our first real-world seller dataset comprises of $64$-dimensional vector embeddings of different products from a large advertisement corpus. Each product/vector is additionally associated with a \emph{seller}, which becomes its color in our setting. There are $20$ million base vectors, around $2500$ sellers, and $5000$ query vectors. This distribution is highly skewed, with an extremely small number (around $7$) of sellers constituting more than $90 \%$ of the data, hence naturally motivating the need for enforcing diversity in the search results. The fraction of products corresponding to the top 20 sellers is shown in \Cref{fig:percentage}.

\item {\em Amazon Automotive Dataset:} Our second real-world dataset is derived from the recently released Amazon dataset~\cite{big-bench}. 
It comprises of $384$-dimensional vector embeddings of product descriptions listed on Amazon under the Automotive category. Each product/vector is additionally associated with a \emph{brand}, which becomes the color. There are around $2$ million base vectors and around $85000$ brands. The distribution, while skewed, is far more balanced than the above seller dataset, with around $10\%$ of the brands accounting for $80\%$ of the vectors as summarized in~\Cref{fig:percentage2}. 

\item{\em ``Skewed'' Semi-synthetic Datasets:} We also consider the publicly available real-world Arxiv dataset~\cite{arxiv} which contains OpenAI embeddings of around 2 million paper abstracts into 1536 dimensional vectors and the classical SIFT dataset of 1M vectors in $128$ dimensions. These datasets do not contain any color information, so we synthetically add this information into the data set.  Specifically, for the Arxiv dataset, we generate the color information as follows: for each vector, with probability $0.9$, we assign a color selected from the set $\{ 1,2,3\}$ uniformly at random, and with $0.1$ probability we  assign a color selected uniformly at random from the set $\{4,\dots, 1000 \}$. Therefore the number of distinct colors is at most $1000$ in this data set. For the SIFT dataset, we sampled one dominant colors with probability $0.8$ and had a uniform distribution over $999$ other colors with probability $0.2$. 

\item{\em ``Balanced'' Semi-synthetic Dataset:} Finally, we also consider another distribution which is globally uniform, but locally skewed. Indeed, we use the same SIFT dataset for the vector data. For colors, we randomly partition the space into one thousand buckets, using a random hyperplane scheme. We then assign a unique \emph{primary color} for each bucket. Now, each vector within any specific bucket is assigned its primary color with a high probability of $0.8$, and a uniformly random non-primary color with the remaining probability. It is then easy to see that the distribution is roughly balanced across colors from a global perspective, but quite skewed in any small local neighborhood.
\end{itemize}

\paragraph{Tasks.}
For all of the above datasets, we seek to retrieve $k=100$ nearest neighbors. In one extreme scenario, we set $k'=1$, i.e., all hundred of the returned results need to be of distinct colors. This type of setting would be relevant in a retrieval augmented generation setting where documents are typically chunked into several parts and each part is vectorized; when the user utters a query, we might want to retrieve the most relevant set of documents (as opposed to the most relevant chunks, which might all be from the same document) to a given query, before passing on these contents to a large-language model which then answers the user query. A natural way to enforce this document-level diversity during retrieval is to treat the document-id of any vector as its color, and using our diverse search routine.

In another scenario which might have more appeal in shopping or advertisement display, we seek to retrieve $k=100$ nearest neighbors, while having $k'=10$, i.e., no more than ten of the $k$ results may belong to any single color. This will promote the retrieval to display a diverse set of sellers/ brands in such scenarios, thereby hopefully increasing user satisfaction and engagement. This can also capture \emph{intent diversity} in regular web-search (when we can use a simple classifier to represent the intent behind the data point as additional meta-data, which becomes the color in our setting -- e.g., car or animal for different images of jaguar, ML or EE concept for transformer, etc.).

\paragraph{Algorithms.} Since our algorithms are enhancements of the DiskANN algorithm, we use that as a natural baseline to compare against.

\begin{itemize}[leftmargin=*]
\item {\em Standard DiskANN Build + Post-Processing (Baseline):} In this baseline, we build a regular DiskANN graph without any diversity constraints. To answer a query, we first invoke the regular DiskANN search algorithm to retrieve $r \gg k$ candidates, again without any diversity constraints. Then we iterate over the retrieved elements in sorted order of distances to the query, and greedily include the ones which do not violate the $k'$ diversity constraint, until we have  $k$ total elements. The parameter $r$ is tunable at search time, and higher $r$ yields more candidates, meaning more diverse candidates can be obtained using the post-processing step. However, higher $r$ also consumes more search complexity. 

\item{\em Standard DiskANN Build + Diverse Search:} In this improvement, we use our diversity-preserving search~\Cref{alg:color_search} discussed in the \Cref{sec:impl}, but the index construction remains the standard DiskANN algorithm.

\item {\em Diverse DiskANN Build + Diverse Search:} For our complete algorithm, we additionally use our diversity-aware index construction ~\Cref{alg:colorindex} (\Cref{sec:impl}) which ensures sufficient edges are present to nodes of different colors in any neighborhood. 
\end{itemize}

\paragraph{Parameter setup.} For all of the above algorithms, we use fairly standard parameters of list-size $L=200$ and graph-degree $64$ when building the graphs. During search, we vary the list-size $L$ at search time to get varying quality search results and plot the recall@100\footnote{Recall@100 is the  size of the intersection of the algorithm's 100 returned results with the true 100 closest diverse candidates, averaged over all queries. The ground-truth set of top 100 diverse NNs for any query can be computed by iterating over all the vectors in sorted order of distances to the query, and greedily including the ones which do not violate the $k'$ diversity constraint, until we have accumulated $k$ total elements.} vs average query latency.

\subsection{Results}\label{sec:results-exp}
Our results are shown in plots of Figures \ref{fig:recall-latency-1}, \ref{fig:recall-latency-10}. As one can see, both of our algorithmic innovations play a crucial role in the overall search quality on the real-world dataset. For example, to achieve $95\%$ recall@100 in the real-world seller dataset, the baseline approach has latencies upwards of $8$ms, while the improved search algorithm brings it down to $\approx 4.5$ms. Making both build and search diverse further brings it down to around $\approx 1.5$ms, resulting in an improvement of 5X. 

The plot in~\Cref{fig:recall-latency-10} reveals an interesting phenomenon: for high recalls (say $90\%$) on the semi-synthetic arXiv dataset, the post-processing approach has a latency of around $90$ms, while the diverse search algorithm when run on the standard graph has a latency of around $135$ms. This is perhaps because the standard graph construction might not have sufficiently many edges between nodes of different colors to ensure that the diverse search algorithm converges to a good local optimum. On the other hand, running the diverse search on the graph constructed keeping diversity in mind during index construction fares the best, with a latency of only around $25$ms. A similar phenomenon occurs in the SIFT semi-synthetic dataset as well.

%\noindent {\bf Build Diversity Parameter Ablation.} 
\subsection{Build Diversity Parameter Ablation}\label{sec:ablation}
In our heuristic graph construction algorithm (see~\Cref{alg:colorindex,alg:colorprune}), the graph edges are added by considering \emph{both the geometry of the vectors and the corresponding colors}. Loosely, the $\alpha$-pruning rule of DiskANN dictates that an edge $(u,v)$ is blocked by an existing edge $(u,w)$ if $d(w,v) \leq d(u,v)/\alpha$. In the original DiskANN algorithm, any edge $(u,v)$ which is blocked is not added. In our setting, we additionally enforce that \emph{an edge needs to be blocked by edges of $m$ different colors} to not be added to the graph, where $m$ is a tuneable parameter. We now perform an ablation capturing the role of $m$ in the graph quality using the skewed SIFT dataset.~\Cref{tbl:build} shows a table with build times for various indices by varying only the $m$ parameter, and~\Cref{fig:recall-latency-ablation} shows the search quality of these different indices.

\begin{table} 
\centering
\begin{tabular}{|c|c|}
\hline
$m$ Parameter  & Build Time (s) \\
\hline
1 &  50 \\
2 &  53 \\
10 &  55 \\
\hline
\end{tabular}
\caption{Build Times w.r.t $m$ Parameter}
\label{tbl:build}
\end{table}

\begin{figure}[!ht]
    \centering
    \begin{minipage}{0.33\textwidth}
        \centering
\includegraphics[width=\textwidth]{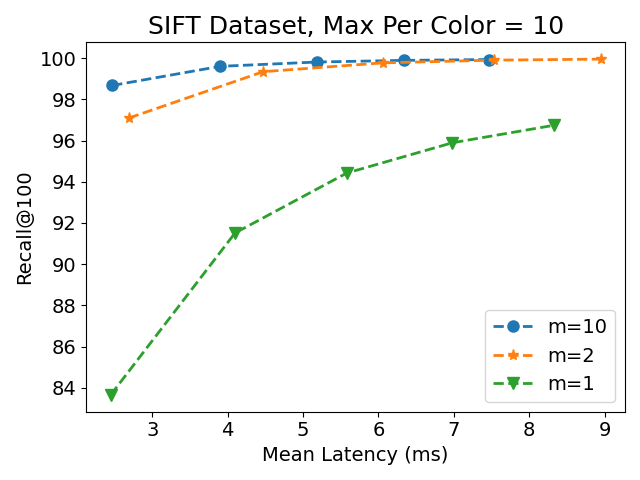}
% \caption{Recall vs Latency for real-world dataset where $k'=10$}
% \label{fig:recall-latency-ads}
    \end{minipage}\hspace{1cm}
    \begin{minipage}{0.33\textwidth}
        \centering
\includegraphics[width=\textwidth]{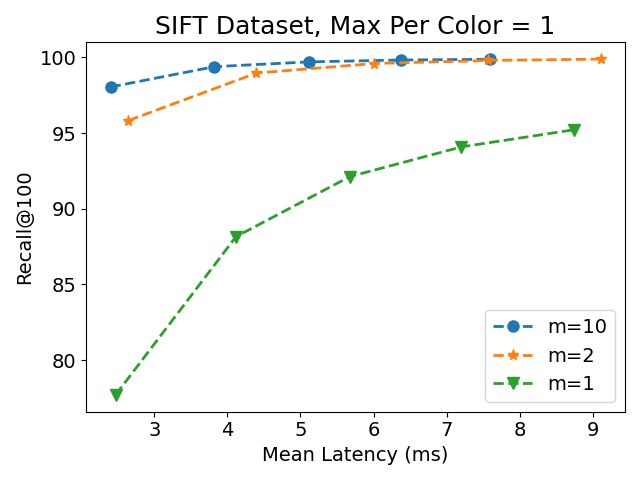}
% \caption{Recall vs Latency for real-world dataset where $k'=1$.}
% \label{fig:recall-latency-arxiv}
    \end{minipage}
    \caption{Recall vs Latency for SIFT-Skewed dataset with $k'=10$ (left) and $k'=1$ (right) by varying the diversity parameter $m$ during index construction. Higher $m$ implies more diversity.}
    \label{fig:recall-latency-ablation}
\end{figure}

%% file: generalized-alg.tex
\section{Algorithms for the General Case} \label{sec:general-algo}
In this section we describe our algorithms for the most general case where $k'$ can take any value between $1$ and $k$, and the diversity metric $\rho$ is an arbitrary diversity metric. First we start with definitions, additional preliminaries, and our main theorem statements \ref{thm:diverse_ann} and \ref{thm:dual_diverse_ann}.

\subsection{Additional Preliminaries, Problem Formulation, and the Main Theorems}

As before, we let $(X,D)$ be the underlying metric space, where $D$ measures the distance between the points. In this section, we assume that we are given a second metric space $(X,\rho)$ which measures the {\em diversity} between the points. As before $P$ is a subset of $n$ points in $X$. So for two points $p_1,p_2\in P$, $\rho(p_1,p_2)$ measures their pairwise diversity.

Again, we use $B_D(p,r)$ (or just $B(p,r)$ for simplicity of exposition) to denote a ball centered at $p$ with radius $r$, i.e.,  $B_D(p,r)=\{u\in X: D(u,p)<r\}$. Similarly, we define the ball $B_\rho(p,r)=\{u\in X: \rho(u,p)<r\}$.

The following definitions recap the discussion in the introduction.

\begin{definition}[$C$-diverse]\label{def:minimal_distance}
Let $S$ be a set of points in $X$. We say $S$ is {\em $C$-diverse} if for any two points $p_1,p_2\in S$, we have $\rho(p_1,p_2)\geq C$. 
\end{definition}

Note that the colorful setting, corresponds to the diversity metric $\rho$ being uniform. That is, we can set $\rho(p_i,p_j)=0$ when $col[p_i]=col[p_j]$, and set $\rho(p_i,p_j)=1$ otherwise. Then, we retrieve the colorful notion of diversity: a set $S$ of size $k$ is colorful iff it is $1$-diverse.
We further generalize this notion to allow the set to contain at most $k'>1$ points that are similar to each other. 

\begin{definition}[$(k',C)$-diverse]\label{def:k'-minimal_distance}
Let $S$ be a set of points in $X$. We say $S$ is $(k',C)$-diverse if for any point $p\in S$, we have $|B_\rho(p,C)\cap S|\le k'$. 
Note that being $(1,C)$-diverse is equivalent to the notion of $C$-diverse.
\end{definition}

We consider two dual variants of the diverse nearest neighbor search problem, both of which use two approximation factors: $c>1$ is the ``dissimilarity'' approximation factor with respect to $D$, and $a>1$ is the ``diversity'' approximation factor with respect to $\rho$.

\begin{definition}[Primal Diverse NN]
    Given a point set $P$, diversity value $C$, and the value $k'\leq k$, the goal of {\em Primal Diverse NN} is to preprocess $P$ and create a data structure such that given a query point $q$, one can quickly return the {\em closest} set $S\subset P$ of size $k$ that is $(k',C)$-diverse. Here closeness of a set $S$ is measured by $S_k$ (See Definition \ref{def:S_i}).

    In the approximate variant, for any $q \in X$, if $\mathsf{OPT}$ is a $(k',C)$-diverse set of $k$ points which minimizes $\mathsf{OPT_k}$, then the data structure outputs $\mathsf{ALG}$ that is $(k',C/a)$-diverse such that $\mathsf{ALG_k} \le c\cdot \mathsf{OPT_k}$.
\end{definition}

\begin{definition}[Dual Diverse NN]
    Given a point set $P$, a radius $R$, and the value $k'\leq k$, the goal of {\em Dual Diverse NN} is to preprocess $P$ and create a data structure such that given a query point $q$, one can quickly return a set $S\subset P$ of size $k$ that lie within the radius $R$, while maximizing the diversity. 

    Formally, for any $q \in X$, let $B_P(q,R)$ be the set of points in $P$ within distance $R$ from $q$, and let $\mathsf{OPT}$ be a $(k',C)$-diverse set of $k^*=\min(k,|B_P(q,R)|)$ points from $B_P(q,R)$ that maximizes $C$. Then the data structure outputs $\mathsf{ALG}$ of size $k^*$ that is $(k',C/a)$-diverse such that $\mathsf{ALG_{k^*}}\le c R$.
\end{definition}

As described in the introduction, the problem addressed in the prior work~\cite{abbar2013diverse} is the dual diverse NN problem, where the only consider $k'=1$.

\paragraph{Results.} Our main theoretical result is captured by the following theorems, which specifies the approximation and running time guarantees for our algorithms solving the primal and dual versions of the diverse nearest neighbor problem. 
%Note that similar to the prior work of \citep{abbar2013diverse}, our result for the dual diverse NN, only holds in the case of $k'=1$.

\begin{theorem}[Primal Diverse ANN] \label{thm:diverse_ann}
Let $\mathsf{OPT}=\{p^*_1,...,p^*_k\}$ be a $(k',C)$-diverse solution that minimizes $\mathsf{\mathsf{OPT_k}}$. Given the graph constructed by Algorithm~\ref{alg:indexing}, the search Algorithm~\ref{alg:general_search} finds a $(k',C/12)$-diverse solution $\mathsf{ALG}$ with $\mathsf{ALG_k}\le \left(\frac{\alpha+1}{\alpha-1}+\epsilon\right)\mathsf{OPT_k}$ in $O\left(k\log_{\alpha}\frac{\Delta}{\epsilon}\right)$ steps, where each step takes $O\left((k^3/k')(8\alpha)^d\log \Delta\right)$ time.
The data structure uses space $O(n(k/k')(8\alpha)^d\log\Delta)$.
\end{theorem}

\smallskip
\begin{theorem}[Dual Diverse ANN] \label{thm:dual_diverse_ann}
Given the graph constructed by Algorithm~\ref{alg:indexing}, the search Algorithm~\ref{alg:dual_search} finds a $(k',C/24)$-diverse NN solution $\mathsf{ALG}$ satisfying $ALG_k\le \left(\frac{\alpha+1}{\alpha-1}+\epsilon\right)\cdot R$ in $\tilde{O}\left((k^4/k')(8\alpha)^d\log^2\frac{\Delta}{\epsilon}\right)$ time, if there exists a $(k',C)$-diverse solution $\mathsf{OPT}$ with $\mathsf{OPT_k}\le R$.
\end{theorem}

\subsection{Algorithm}
\noindent\textbf{The preprocessing algorithm.} The indexing algorithm, which is the same for both the primal and dual versions of the problem, is shown in Algorithm~\ref{alg:indexing}. Line 12 of the algorithm uses the greedy algorithm of \cite{gonzalez1985clustering}, defined below.

\begin{algorithm}
\caption{Indexing algorithm for diverse NN}
\label{alg:indexing}
\begin{algorithmic}[1]
\STATE \textbf{Input}: A set of $n$ points $P=\{p_1,...,p_n\}$; $k$ is the size of the output; $k'$ is the parameter in the diversity definition; $\alpha$ is the parameter used for pruning.
\STATE \textbf{Output}: A directed graph $G=(V,E)$ where $V=\{1,...,n\}$ are associated with the point set $P$.
\FOR{each point $p\in P$}
    \STATE Sort all points $u\in P$ based on their distance from $p$ and put them in a list $\mathcal{L}$ in that order
    \WHILE{$\mathcal{L}$ is not empty}
        \STATE $u\gets \argmin\limits_{u\in \mathcal{L}}D(u,p)$
        \STATE Initialize $\mathsf{bag}[u]\gets \{u\}$
        \FOR{each point $v\in \mathcal{L}$ in order}
            \IF{$D(u,v)\le D(p,u)/(2\alpha)$}
                \STATE $\mathsf{bag}[u]\gets \mathsf{bag}[u]\cup v$
                \STATE remove $v$ from $\mathcal{L}$
            \ENDIF
        \ENDFOR
        \STATE $\mathsf{rep}[u]\gets$ use the greedy algorithm of Gonzales to choose $k/k'$ points in $\mathsf{bag}[u]$ to approximately maximize the minimum pairwise diversity.
        \STATE add edges from $p$ to $\mathsf{rep}[u]$
        \STATE Remove $u$ from $\mathcal{L}$
    \ENDWHILE
\ENDFOR
\end{algorithmic}
\end{algorithm}

\noindent\textbf{Gonzales' greedy algorithm.} Given a set of $n$ points and a parameter $m$, the algorithm picks $m$ points as follows. The first point is chosen arbitrarily. Then, in each of the next $m-1$ steps, the algorithm picks the point whose minimum distance w.r.t. $\rho$ to the currently chosen points is maximized. It is known \cite{gonzalez1985clustering} that this algorithm provides a $2$-approximation for the problem of picking a subset of size $m$ which maximizes the minimum pairwise diversity distance between the picked points. Moreover, the picked set has an anti-cover property which we will discuss in Proposition \ref{prop:greedy}.

\noindent\textbf{Primal Search Algorithm.} Algorithm~\ref{alg:general_search} shows the search algorithm for the primal version of diverse nearest neighbor. 
%\xnote{are we removing the special analysis for $k'=1$}
%The algorithm has a different condition for $k'=1$ as in this case we can slightly improve the performance of the algorithm as shown in Section \ref{sec:improved-analysis}. The general case of 
The algorithm is analyzed in Section \ref{sec:primal-analysis}. The initialization step of line 3, can be done using the following algorithm.

%\xnote{perhaps remove line 6,7, and 8 if we don't have the $k'=1$ analysis}

\begin{algorithm}
\caption{Search algorithm for primal diverse NN}
\label{alg:general_search}

\begin{algorithmic}[1]
\STATE \textbf{Input}: A graph $G=(V,E)$ with $N_{out}(p)$ denoting the out edges of $p$; query $q$, number of optimization steps $T$; diversity lower bound $C$.
\STATE \textbf{Output}: A set of $k$ points $\mathsf{ALG}$.
\STATE Initialize $\mathsf{ALG}=\{p_1,...,p_k\}$ to be a set of $k$ points that are $(k',C/12)$-diverse using the initialization step proved in Lemma \ref{lm:initialization}.
\FOR{$i=1$ to $T$}
    % \State $U\gets \{u | (p_j,u)\in E\ s.t. \exists\ p_j\in ALG\}$
    \STATE $U\gets \bigcup\limits_{p\in \mathsf{ALG}}(N_{out}(p)\cup p)$, and sort $U$ based on their distance from $q$
%    \IF{$k'=1$}
%        \STATE $\mathsf{ALG}=\varnothing$
%    \ELSE
        % \State $ALG^t=ALG^{t-1}\setminus p_k$
        \STATE $\mathsf{ALG}\gets$ the closest $k-1$ points in $\mathsf{ALG}$
%    \ENDIF
    
    \FOR{each point $u \in U$ in order}
        \IF{$\mathsf{ALG}\bigcup u$ is $(k',C/12)$-diverse}
            \STATE $\mathsf{ALG}\gets \mathsf{ALG}\cup u$
        \ENDIF
        \IF{$|\mathsf{ALG}|=k$}
            \STATE Break
        \ENDIF
    \ENDFOR
    % \If{$\mathsf{ALG}^{t}\ge \mathsf{ALG}^{t-1}$}
    % \State Break
    % \EndIf
\ENDFOR
\STATE \textbf{Return} $\mathsf{ALG}$
\end{algorithmic}
\end{algorithm}

\noindent\textbf{The initialization step.} Given a set $P$ of $n$ points equipped with metric distance $\rho$, and parameters $k'$ and $k$, and lower bound diversity $C$, the goal is to pick a subset $S\subseteq P$ of size $k$ which is $(k',C/4)$ diverse or otherwise output that no $(k',C)$-diverse subset $S$ exists. We use the following algorithm
\begin{itemize}
    \item Initialize $\mathsf{SOL}=\emptyset$
    \item While there exists a point $p\in P$ such that the ball $B=B_\rho(p,C/4)$ has $k'$ points in it, (i.e., $|B\cap P|>k'$), 
    \begin{itemize}
        \item Add an arbitrary subset of $B\cap P$ of size $k'$ to $\mathsf{SOL}$. 
        \item Remove all points in $2B=B_\rho(p,C/2)$ from $P$.
    \end{itemize}
    \item Add all remaining points in $P$ to $\mathsf{SOL}$.
    \item If $|\mathsf{SOL}|\geq k$, return an arbitrary subset of it of size $k$, otherwise return `no solution'.
\end{itemize}

\begin{lemma}[Initialization]\label{lm:initialization}
If $P$ has a subset $\mathsf{OPT}$ of size $k$ that is $(k',C)$-diverse, our initialization algorithm finds a $(k',C/4)$-diverse subset of size $k$.
\end{lemma}
\begin{proof}
    Note that it is straightforward to see why the set $\mathsf{SOL}$ that we get at the end is $(k',C/4)$-diverse. This is because first of all, each time we pick $k'$ points in a ball $B$ and add them to $\mathsf{SOL}$, we make sure that no additional point will ever be picked in $2B$ and thus within distance $C/4$ of the points we pick there will be at most $k'$ points in the end. Second, at the end, every remaining ball of radius $C/4$ has less than or equal to $k'$ points in it. Therefore, we can pick all such points in the solution and everything we picked will be $(k',C/4)$ diverse.

    Next we argue that we are in fact able to pick at least $k$ points in total which completes the argument. We do it by following the procedure of our algorithm and comparing it with $\mathsf{OPT}$. 
    At each iteration of the while loop that we remove $P\cap 2B$, we add exactly $k'$ points from $P\cap 2B$ to our solution $\mathsf{SOL}$. 
    Now note that the optimal solution $\mathsf{OPT}$ cannot have more than $k'$ points in $2B$ because by triangle inequality any pair of points in $2B$ have distance at most $C$, and picking more than $k'$ points in this ball contradicts the fact that $\mathsf{OPT}$ is $(k',C)$ diverse. Thus we can have an one-to-one mapping from each point in $\mathsf{OPT}\cap 2B$ to the $k'$ points in $P\cap 2B$ added to $\mathsf{SOL}$.
    At the end of the while iteration, we know any unmapped point in $\mathsf{OPT}$ still exists in $P$, so we just map it to itself. By doing this, we can have an one-to-one mapping from $\mathsf{OPT}$ to $\mathsf{SOL}$, which means that $|\mathsf{SOL}|\ge|\mathsf{OPT}|=k$. 
    % Thus the removal of $P\cap 2B$ can only remove at most $k'$ points from $OPT$.
    % Therefore at the end of the while iteration, the size of the remaining points in $P$, which should be at least the number of points remaining in $OPT$, is more than the number of points that we still need to pick. Therefore, our algorithm will always be able to pick $k$ points assuming such an $OPT$ exists.
\end{proof}

\noindent\textbf{Dual Search Algorithm.} Algorithm~\ref{alg:dual_search} shows the search algorithm for the dual version of the diverse nearest neighbor problem. We provide the analysis in Section \ref{sec:dual-analysis}.

\subsection{Analysis of the Primal Diverse NN Algorithm}\label{sec:primal-analysis}
In this section, we prove Theorem~\ref{thm:diverse_ann} that gives the approximation and running time guarantees for Algorithm~\ref{alg:indexing} and Algorithm~\ref{alg:general_search}.

\iffalse

\begin{corollary}\label{cor:colorful_ann}
Let $OPT=\{p^*_1,...,p^*_k\}$ be the $k'$-colorful solution with minimized $\mathsf{OPT_k}$. Given the graph constructed by Algorithm~\ref{alg:indexing}, the search Algorithm~\ref{alg:general_search} finds a $k'$-colorful solution $\mathsf{ALG}$ with $\mathsf{ALG_k}\le \left(\frac{\alpha+1}{\alpha-1}+\epsilon\right)\mathsf{OPT_k}$ in $O\left(k\log_{\alpha}\frac{\Delta}{\epsilon}\right)$ steps where each step takes $O\left((k^2/k')(8\alpha)^d\log \Delta\right)$ time.
\end{corollary}

\begin{proof} 
By setting $\rho(p_i,p_j)=1$ for $col[p_i]\neq col[p_j]$ and $0$ otherwise, a set $P$ is $k'$-colorful if and only if it is $(k',0.05C)$-diverse. Then we can apply Theorem~\ref{thm:diverse_ann} to prove this corollary.
\end{proof}
\fi

\begin{lemma}\label{lm:degree_diverse}
The graph constructed by Algorithm~\ref{alg:indexing} has degree limit $O((k/k')(8\alpha)^d\log\Delta)$.
\end{lemma}

\begin{proof}
Let's first bound the number of points not removed by others, then according to Line 14-15 in Algorithm~\ref{alg:indexing}, the degree bound will be that times $k/k'$.

We use $\mathsf{Ring}(p,r_1,r_2)$ to denote the points whose distance from $p$ is larger than $r_1$ but smaller than $r_2$. For each $i\in [\log_2 \Delta]$, we consider the $\mathsf{Ring}(p,D_{max}/2^i,D_{max}/2^{i-1})$ separately. According to Lemma~\ref{lm:doubling_dimension}, we can cover $\mathsf{Ring}(p,D_{max}/2^i,D_{max}/2^{i-1})\cap P$ using at most $m\le O((8\alpha)^d)$ small balls with radius $\frac{D_{max}}{2^{i+2}\alpha}$. According to the pruning criteria in Line 9, within each small ball, there will be at most one point remaining. This establishes the degree bound of $O((k/k')(8\alpha)^d\log\Delta)$.
% We use $Ring(p,r_1,r_2)$ to denote the points whose distance from $p$ is larger than $r_1$ but smaller than $r_2$. For each $i\in [\log_2 \Delta]$, we consider the $Ring(p,D_{max}/2^i,D_{max}/2^{i-1})$ separately. According to Lemma~\ref{lm:doubling_dimension}, we can cover $Ring(p,D_{max}/2^i,D_{max}/2^{i-1})\cap P$ using at most $m\le O((4\alpha)^d)$ small balls with radius $\frac{D_{max}}{2^{i+1}\alpha}$. Now, we argue that within any such small ball $B$, the number of points connected to $p$ is at most $k/k'$. First, we note that for any two points $u,v$ within $B$, we have $D(u,v)\le D(p,v)/\alpha$. Therefore, any two points $u,v\in N_{out}(p)\cap B$ must satisfy $\rho(u,v)>0.4C$. Otherwise the later one should have been removed. Now, for the sake of contradiction, suppose there are $m>k/k'$ points $z_1,...,z_m$ within $N_{out}(p)\cap B$. We will bound the size of $\mathsf{rep}[z_m]$ for the last point to join $N_{out}(p)$. For any $z_i$ with $i<m$, either $z_i\in \mathsf{rep}[z_m]$ or there exists another $z'_i\in \mathsf{rep}[z_m]$ with $\rho(z_i,z'_i)<0.2C$. For any different $i,j$ which don't belong to $\mathsf{rep}[z_m]$, because $\rho(z_i,z_j)>0.4C$, by triangle inequality, their $z'_i$ and $z'_j$ must be different. Therefore, the size of $\mathsf{rep}[z_m]$ is lower bounded by $m-1\ge k/k'$, and $z_m$ would have been removed. Given that $|N_{out}(p)\cap B|\le k/k'$ for any ball $B$, a similar counting argument gives the degree upper bounded of $O((k/k')(4\alpha)^d\log\Delta)$
\end{proof}

\begin{lemma}\label{lm:p_star_exists_diverse}
Suppose $\mathsf{OPT}=\{p^*_1,...,p^*_k\}$ is a $(k',C)$-diverse solution with minimized $\mathsf{OPT_k}$, and let $\mathsf{ALG}=\{p_1,...,p_k\}$ be the current solution (ordered by distance from $q$). If $p_k\notin \mathsf{OPT}$, there exists a point $p^*\in \mathsf{OPT}\setminus \mathsf{ALG}$ such that $|B_{\rho}(p^*,C/2)\cap (\mathsf{ALG}\setminus p_k)|<k'$ and $\mathsf{ALG}\setminus p_k\bigcup p^*$ is $(k',C/4)$-diverse.
\end{lemma}

\begin{proof}
% We use $B_{\rho}(p,r)$ to denote the ball in the $(X,\rho)$ metric space. Consider the set of balls $B_{\rho}(p^*_1,C/2),...,B_{\rho}(p^*_k,C/2)$. They are disjoint because $\rho(p^*_i,p^*_j)\ge C$ for $i\neq j$. Then there must exists a $p^* \in \mathsf{OPT}$ such that $B_{\rho}(p^*,C/2)\bigcap (\mathsf{ALG}\setminus p_k)=\varnothing$ by a counting argument. We can also check that for this $p^*$, $\mathsf{ALG}\setminus p_k\bigcup p^*$ is $(k,C/16)$ diverse.

Recall that we use $B_{\rho}(p,r)$ to denote the ball in the $(X,\rho)$ metric space. Because $p_k\notin \mathsf{OPT}$, we have $\overline{\mathsf{OPT}}=\mathsf{OPT}\setminus \mathsf{ALG}\neq \varnothing$. We repeatedly perform the following operation until $\overline{\mathsf{OPT}}$ gets empty: select a point $p$ from $\overline{\mathsf{OPT}}$, and let $z=B_{\rho}(p,C)\cap \overline{\mathsf{OPT}}$, and remove $z$ from $\overline{\mathsf{OPT}}$. By doing this, we can get a list of points $\{p^*_1,...,p^*_m\}$ and a partition of $\mathsf{OPT}\setminus \mathsf{ALG}=z_1\cup z_2...\cup z_m$. By definition, we have the following properties:
%$$\{p^*_1,...,p^*_m\}\cap \mathsf{ALG}=\varnothing,~~ z_i\cap z_j=\varnothing \text{ for } i\neq j,~~\sum_i |z_i|=|\mathsf{OPT}\setminus \mathsf{ALG}|=|\mathsf{ALG}\setminus \mathsf{OPT}|~. $$

\begin{itemize}
    \item $\{p^*_1,...,p^*_m\}\cap \mathsf{ALG}=\varnothing$
    \item $z_i\cap z_j=\varnothing$ for $i\neq j$
    \item $\sum_i |z_i|=|\mathsf{OPT}\setminus \mathsf{ALG}|=|\mathsf{ALG}\setminus \mathsf{OPT}|$
\end{itemize}

Now let $w_i=B_{\rho}(p^*_i,C/2)\cap (\mathsf{ALG}\setminus p_k\setminus \mathsf{OPT})$. Because all the $B_{\rho}(p^*_i,C/2)$ balls are disjoint, $\sum_i |w_i|\le |\mathsf{ALG}\setminus p_k\setminus \mathsf{OPT}|<|\mathsf{OPT}\setminus \mathsf{ALG}|=\sum_i |z_i|$, there must exist an $i$ such that $|w_i|<|z_i|$. For that $i$, we have that $|B_{\rho}(p^*_i,C/2)\cap (\mathsf{ALG}\setminus p_k)|$ is equal to
\begin{align}
= &|B_{\rho}(p^*_i,C/2)\cap (\mathsf{ALG}\cap \mathsf{OPT})|+|B_{\rho}(p^*_i,C/2)\cap (\mathsf{ALG}\setminus p_k \setminus \mathsf{OPT})|\tag{Because $p_k\notin \mathsf{OPT}$}\\
=&|B_{\rho}(p^*_i,C/2)\cap (\mathsf{ALG}\cap \mathsf{OPT})|+|w_i|\notag \\
<&|B_{\rho}(p^*_i,C/2)\cap (\mathsf{ALG}\cap \mathsf{OPT})|+|z_i|\notag\\
\le&|B_{\rho}(p^*_i,C/2)\cap (\mathsf{ALG}\cap \mathsf{OPT})|+|B_{\rho}(p^*_i,C)\cap (\mathsf{OPT}\setminus\mathsf{ALG})|\notag \\
\le&|B_{\rho}(p^*_i,C)\cap (\mathsf{ALG}\cap \mathsf{OPT})|+|B_{\rho}(p^*_i,C)\cap (\mathsf{OPT}\setminus\mathsf{ALG})|\notag \\
=&|B_{\rho}(p^*_i,C)\cap \mathsf{OPT}|\notag \\
\le&  k'\notag
\end{align}

Therefore, we get $B_{\rho}(p^*_i,C/2)\cap(\mathsf{ALG}\setminus p_k)<k'$. Now, for any point $p\in B_{\rho}(p^*_i,C/4)$, $|B_{\rho}(p,C/4)\cap (\mathsf{ALG}\setminus p_k)|\le |B_{\rho}(p^*_i,C/2)\cap (\mathsf{ALG}\setminus p_k)|<k'$, so we know that $\mathsf{ALG}\setminus p_k\cup p^*_i$ is $(k',C/4)$-diverse.
% Because for any $p^*_i\in \mathsf{OPT}$, $|B_{\rho}(p^*_i,C)\cap \mathsf{OPT} |\le k'$. We can extract at least $k/k'$ points $\{z^*_1,...,z^*_{k/k'}\}\subseteq P^*$ so that $B_{\rho}(z^*_i,C/2)$ are disjoint for different $z^*_i$. By a counting argument, there must exist a $p^*\in \{z^*_1,...,z^*_{k/k'}\} \subseteq \mathsf{OPT}$ such that $|B_{\rho}(p^*,C/2)\bigcap (\mathsf{ALG}\setminus p_k)|<k'$. To show that $ALG'=\mathsf{ALG}\setminus p_k\bigcup p^*$ is $(k',0.05C)$-diverse, for any $p\in ALG'$ with $\rho(p,p^*)\le 0.05C$, we have $B_{\rho}(p,0.05C)\bigcap ALG'\subseteq B_{\rho}(p^*,C/2)\bigcap ALG'\le k'$. Therefore $ALG'=\mathsf{ALG}\setminus p_k\bigcup p^*$ is $(k',0.05C)$-diverse.

\end{proof}
The following is the well-known anti-cover property of the greedy algorithm of Gonzales whose proof we include for the sake of completeness.
\begin{proposition}\label{prop:greedy}
In Line 14 of Algorithm~\ref{alg:indexing}, let $\mathsf{rep}[u]$ be the output of greedily choosing $k/k'$ points in $\mathsf{bag}[u]$ maximizing minimum pairwise diversity. If a point $p\in \mathsf{bag}[u]\setminus \mathsf{rep}[u]$, we have $\min\limits_{v\in \mathsf{rep}[u]}\rho(p,v)\le \min\limits_{v_1,v_2\in \mathsf{rep}[u]}\rho(v_1,v_2)$.
\end{proposition}
\begin{proof}
For the sake of contradiction, suppose $\min\limits_{v\in \mathsf{rep}[u]}\rho(p,v)>\min\limits_{v_1,v_2\in \mathsf{rep}[u]}\rho(v_1,v_2)$, and the pairwise diversity minimizer is achieved by $\min\limits_{v_1,v_2\in \mathsf{rep}[u]}\rho(v_1,v_2)=\rho(x,y)$. Without loss of generality, we assume $x$ is added to $\mathsf{rep}[u]$ before $y$. At the time step $t$ when $y$ was added to $\mathsf{rep_t}[u]$, $\min\limits_{v\in \mathsf{rep_t}[u]}\rho(y,v)=\rho(x,y)$ and $\min\limits_{v\in \mathsf{rep_t}[u]}\rho(p,v)\ge \min\limits_{v\in \mathsf{rep}[u]}\rho(p,v)> \rho(x,y)$, so $y$ wouldn't have been chosen by the greedy algorithm. Therefore, we have derived a contradiction.
\end{proof}

\begin{lemma}\label{lm:update_diverse}
There always exists a point $p'$ connected from some point $w\in \mathsf{ALG}$ such that
\begin{enumerate}
\item $\mathsf{ALG}\setminus p_k \bigcup p'$ is $(k',C/12)$-diverse 
\item $D(p',q)\le D(p_k,q)/\alpha+\mathsf{OPT_k}(1+1/\alpha)$
\end{enumerate}
\end{lemma}

\begin{proof}
According to Lemma~\ref{lm:p_star_exists_diverse}, for any current solution $\mathsf{ALG}$ with $p_k\notin \mathsf{OPT}$, there exists a point $p^*\in \mathsf{OPT}\setminus \mathsf{ALG}$ such that $\mathsf{ALG}\setminus p_k\cup p^*$ is $(k',C/4)$-diverse. Let $w\in \mathsf{ALG}$ be the closest point to $p^*$. If there exists an edge from $w$ to $p^*$, replacing $p_k$ with $p^*$ is a potential update. We set $p'=p^*$ and $D(p',q)\le \mathsf{OPT_k}$ satisfies the distance upper bound above.

% In the following, we consider two cases depending on how the edge $(w,p^*)$ is removed.

% \begin{enumerate}
% \item $(w,p^*)$ was removed because of Condition 1: In this case, there exists another point $p'$ with $D(p^*,p')\le D(w,p^*)/\alpha$ and $\rho(p^*,p')\le 0.4C$. Because $|B_{\rho}(p^*,C/2)\bigcap (\mathsf{ALG}\setminus p_k)|<k'$, we have $|B_{\rho}(p',0.1C)\bigcap (\mathsf{ALG}\setminus p_k)|\subseteq |B_{\rho}(p^*,C/2)\bigcap (\mathsf{ALG}\setminus p_k)|<k'$, so the addition of such $p'$ satisfies that $\mathsf{ALG}\setminus p_k\cup p'$ is $(k',0.05C)$-diverse. \snote{again this is using a triangle inequality argument, right?} \xnote{yes. if the new added point satisfies $B_{\rho}(p,2C)\cap P<k'$, then the addition of $p$ is still $(k',C)$-diverse}
% \item $(w,p^*)$ was removed because of Condition 2: In this case, there exists $z_1,...,z_{k/k'}\in B(p^*,D(w,p^*)/\alpha)$ all with diversity distance at least $0.2C$ from each other. Therefore, for any $p_i\in \mathsf{ALG}\setminus p_k$, there can't exist two $z_j$ and $z_{j'}$ s.t. $\rho(p_i,z_j)<0.1C$ and $\rho(p_i,z_{j'})<0.1C$. By a counting argument, we can find at least one $z_i$ s.t. $|B_{\rho}(z_i,0.1C)\cap (\mathsf{ALG}\setminus p_k)|<k'$. Because $w$ is the closest point from $p^*$, we know that $\{z_1,...,z_{k/k'}\}\cap \mathsf{ALG}=\varnothing$, so we can let $p'=z_i$ and then $\mathsf{ALG}\setminus p_k \cup p'$ is $(k',0.05C)$-diverse.
% \end{enumerate}

Otherwise, we let $u$ be the point where $p^*\in \mathsf{bag}[u]$ but not selected into $\mathsf{rep}[u]$. For any point $p'\in \mathsf{bag}[u]$, $D(p',u)<D(w,u)/(2\alpha)$, so $D(p',p^*)<D(w,u)/\alpha<D(w,p^*)$. This means that all points in $\mathsf{bag}[u]$ are closer to $p^*$ than $w$, so they can't belong to $\mathsf{ALG}$. In the following, we consider two cases depending on whether $\min\limits_{v\in \mathsf{rep}[u]}\rho(p^*,v)\ge C/3$. In each case, we will find a desired $p'\in \mathsf{rep}[u]$ and it is connected to $w$.

\begin{enumerate}
\item $\min\limits_{v\in \mathsf{rep}[u]}\rho(p^*,v)<C/3$: In this case, there exists another point $p'\in \mathsf{rep}[u]$ with $D(p^*,p')\le D(p^*,u)+D(u,p')\le D(w,u)/\alpha$ and $\rho(p^*,p')<C/3$. Because $|B_{\rho}(p^*,C/2)\bigcap (\mathsf{ALG}\setminus p_k)|<k'$, we have $|B_{\rho}(p',C/6)\bigcap (\mathsf{ALG}\setminus p_k)|\subseteq |B_{\rho}(p^*,C/2)\bigcap (\mathsf{ALG}\setminus p_k)|<k'$, so the addition of such $p'$ satisfies that $\mathsf{ALG}\setminus p_k\cup p'$ is $(k',C/12)$-diverse. 
% \snote{should we mention that there is an edge as well?}\xnote{This argument appeared before at the end of Lemma~\ref{lm:p_star_exists_diverse}}\xnote{add explicit edge}
\item $\min\limits_{v\in \mathsf{rep}[u]}\rho(p^*,v)\ge C/3$: In this case, according to Proposition~\ref{prop:greedy}, we have $\mathsf{rep}[u]=\{z_1,...,z_{k/k'}\}\subseteq B(u,D(u,w)/(2\alpha))$ all with diversity distance at least $C/3$ from each other. Therefore, for any $p_i\in \mathsf{ALG}\setminus p_k$, there can't exist two $z_j$ and $z_{j'}$ s.t. $\rho(p_i,z_j)<C/6$ and $\rho(p_i,z_{j'})<C/6$. By a counting argument, we can find at least one $z_i$ s.t. $|B_{\rho}(z_i,C/6)\cap (\mathsf{ALG}\setminus p_k)|<k'$. Finally, we let $p'=z_i$ where $\mathsf{ALG}\setminus p_k \cup p'$ is $(k',C/12)$-diverse.
\end{enumerate}

We have proved that the $p'$ we found satisfies the $(k',C/12)$-diverse criteria. Now we will bound its distance upper bound.
\begin{align}
D(p',q)&\le D(p^*,q)+D(p',p^*) \le D(p^*,q)+D(p',u)+D(p^*,u) \notag\\
&\le D(p^*,q)+D(w,u)/(2\alpha)+D(w,u)/(2\alpha) \tag{Line 9 in Algorithm~\ref{alg:indexing}}\\
&\le D(p^*,q)+D(w,u)/\alpha \notag\\
&\le D(p^*,q)+D(w,p^*)/\alpha \tag{Because $u$ is ordered earlier than $p^*$}\\
&\le D(p^*,q)+D(w,q)/\alpha+D(p^*,q)/\alpha 
\le D(p_k,q)/\alpha+\mathsf{OPT_k}(1+1/\alpha)\notag
\end{align}

\end{proof}

\begin{proof}[Proof of Theorem~\ref{thm:diverse_ann}]
% By Lemma~\ref{lm:p_star_exists_diverse} and Lemma~\ref{lm:update_diverse}, we know that at each time step, we can replace the current farthest point $p_k$ with a new point $p'$ satisfying $(k',0.05C)$-diverse and $D(p',q)\le D(p_k,q)/\alpha+D^*_k(1+1/\alpha)$. Following the same proof argument as in Theorem~\ref{thm:colorful_diverse_ann}, we can get the same convergence bound.

Regarding the running time, the total number of edges connected from any point in $\mathsf{ALG}$ is bounded by $|U|\le O((k^2/k')(8\alpha)^d\log \Delta)$. In each step, the algorithm first sorts all these edges and then checks whether each of them can be added to the new $\mathsf{ALG}$ set. The total time spent per step is $O(k|U|+|U|\log|U|)$. Usually, we assume $k\gg \log|U|$, and we can have the overall time complexity to be $O\left((k^3/k')(8\alpha)^d\log\Delta\right)$ per step.

To analyze the approximation ratio, at time step $t$, we use $\mathsf{ALG^t}=\{p^t_1,...,p^t_k\}$ to denote the current {\em unordered} solution. We denote $\mathsf{ALG^t_k}=\max\limits_{i\in[k]}D(p^t_i,q)$. According to Algorithm~\ref{alg:general_search} and Lemma~\ref{lm:update_diverse}, if $p_i$ is updated at time step $t$, we have $D(p^t_i,q)\le D(p^{t-1}_i,q)/\alpha+\mathsf{OPT_k}(1+1/\alpha)$. By an induction argument, if a point $p_i$ is updated by $t$ times at the end of time step $T$, we have $D(p^T_i,q)\le \frac{D(p^0_i,q)}{\alpha^t}+\frac{\alpha+1}{\alpha-1}\mathsf{OPT_k}$. 

We now prove that $\mathsf{ALG^T_k}\le \max\limits_{i}\frac{D(p^0_i,q)}{\alpha^{T/k}}+\frac{\alpha+1}{\alpha-1}\mathsf{OPT_k}$. Let $i\in[k]$ be the index achieving the maximal distance upper bound. For the sake of contradiction, if $\mathsf{ALG^T_k}>\frac{D(p^0_i,q)}{\alpha^{T/k}}+\frac{\alpha+1}{\alpha-1}\mathsf{OPT_k}$, this means that $p^T_i$ was updated for at most $T/k-1$ times. By a counting argument, there exists another index $j$ which was updated for at least $T/k+1$ times. However, at the time $t$ when $p^t_j$ was already updated for $T/k$ times, $D(p^t_j,q)\le \frac{D(p^0_j,q)}{\alpha^{T/k}}+\frac{\alpha+1}{\alpha-1}\mathsf{OPT_k} < \mathsf{ALG^T_k}\le \mathsf{ALG^t_k}$, so the algorithm wouldn't have chosen $p^t_j$ to optimize cause it couldn't have had the maximal distance at that time, leading to a contradiction. Therefore, we prove that $\mathsf{ALG^T_k}\le \max\limits_{i}\frac{D(p^0_i,q)}{\alpha^{T/k}}+\frac{\alpha+1}{\alpha-1}\mathsf{OPT_k}$.

%For the sake of contradiction, suppose after $T$ steps, $D^T>\frac{D(p^0_i,q)}{\alpha^{T/(k\log\alpha)}}+\frac{\alpha+1}{\alpha-1}D^*_k$, which means that the index $i$ achieves the maximal $D(p^T_i,q)$ was updated for fewer than $T/(k\log\alpha)$ times. By a counting argument, some other index $j$ was updated for at least $T/k$ times. However, by the time index $j$ was updated for $T/k$ times, its distance to $q$ had already been upper bounded by $\frac{D(p^0_j,q)}{\alpha^{T/k}}+\frac{\alpha+1}{\alpha-1}D^*_k$.

Now we consider the following three cases depending on the value of the maximal $D(p^0_i,q)$. The case analysis here is similar to the proof in Theorem 3.4 from~\cite{indykxu2024worst}.
\begin{enumerate}
\item[Case 1:] $D(p^0_i,q)>2D_{max}$. Let $p^*_k$ be the point having the maximal distance from $q$ in an optimal solution $\mathsf{OPT}$. We know that for any $p^0_i$, we have $D(p^*_k,q)\ge D(p^0_i,q)-D(p^0_i,p^*_k)\ge D(p^0_i,q)-D_{max}\ge D(p^0_i,q)/2$. Therefore, the approximation ratio after $T$ optimization steps is upper bounded by $\frac{\mathsf{ALG^T_k}}{D(p^*_k,q)}\le \frac{D(p^0_i,q)}{D(p^*_k,q)\alpha^{T/k}}+\frac{\alpha+1}{\alpha-1}\le \frac{2}{\alpha^{T/k}}+\frac{\alpha+1}{\alpha-1}$. A simple calculation shows that we can get a $(\frac{\alpha+1}{\alpha-1}+\epsilon)$ approximate solution in $O(k\log_{\alpha}\frac{2}{\epsilon})$ steps.

\item[Case 2:] $D(p^0_i,q)\le 2D_{max}$ and $\mathsf{OPT_k}>\frac{\alpha-1}{4(\alpha+1)}D_{min}$. To satisfy $\frac{D(p^0_i,q)}{\alpha^{T/k}}+\frac{\alpha+1}{\alpha-1}\mathsf{OPT_k}\le (\frac{\alpha+1}{\alpha-1}+\epsilon)\mathsf{OPT_k}$, we need $\frac{D(p^0_i,q)}{\alpha^{T/k}}\le \epsilon \mathsf{OPT_k}$. Applying the lower bound $\mathsf{OPT_k}\ge \frac{\alpha-1}{4(\alpha+1)}D_{min}$, we can get that $T\ge k\log_{\alpha}\frac{2(\alpha+1)\Delta}{(\alpha-1)\epsilon}$ suffices.

\item[Case 3:] $D(p^0_i,q)\le 2D_{max}$ and $\mathsf{OPT_k}\le \frac{\alpha-1}{4(\alpha+1)}D_{min}$. In this case, we must have $k=1$, because otherwise $D(p^*_k,p^*_1)\le 2D(p^*_k,q)<D_{min}$,violating the definition of $D_{min}$. Suppose $k=1$ and the problem degenerates to the standard nearest neighbor search problem. After $T$ optimization steps, if $p^T_1$ is still not the exact nearest neighbor, we have $D(p^T_1,q)\ge D(p^T_1,p^*_1)-\mathsf{OPT_1}\ge \frac{D_{min}}{2}$. Applying the upper bound of $D(p^T_1,q)$ and $\mathsf{OPT_1}$, we have $\frac{D_{min}}{2}\le D(p^T_1,q)\le \frac{D(p^0_1,q)}{\alpha^{T}}+\frac{\alpha+1}{\alpha-1}\mathsf{OPT_1}\le \frac{D(p^0_1,q)}{\alpha^{T}}+\frac{D_{min}}{4}$. This can happen only if $T\le \log_{\alpha}\frac{\Delta}{8}$.
\end{enumerate}
\end{proof}

%%%%%%%%%%%%%%%%%%%%%%%%%%%%%%%%%%%%%%%5
%.................DUAL.................%
%%%%%%%%%%%%%%%%%%%%%%%%%%%%%%%%%%%%%%%%
\subsection{Analysis for the Dual Diverse NN Algorithm}\label{sec:dual-analysis}

In this section we analyze Algorithm~\ref{alg:dual_search}.

\begin{algorithm}
\caption{Search algorithm for dual diverse NN}
\label{alg:dual_search}
\begin{algorithmic}[1]
\STATE \textbf{Input}: A graph $G=(V,E)$ with $N_{out}(p)$ denoting the out edges of $p$; query $q$; distance bound $R$; distance approximation error $\epsilon$.
\STATE  \textbf{Output}: A set of $k$ points $\mathsf{ALG}$.
% \STATE  $\mathsf{ALG}\gets \{p_1,...,p_k\}$ picked by the greedy algorithm of Gonzales for approximately maximizing the minimum pairwise diversity.
% \STATE  $\overline{C}\gets 4\min\limits_{p_i,p_j\in \mathsf{ALG}}\rho(p_i,p_j)$
\STATE Use binary search to find a maximal $C$ such that the initialization step proved in Lemma \ref{lm:initialization} outputs a $(k',C)$-diverse set $\mathsf{ALG}=\{p_1,...,p_k\}$
\STATE $\overline{C}\gets 4C$

\WHILE{$\max\limits_{p\in \mathsf{ALG}} D(p,q)>(\frac{\alpha+1}{\alpha-1}+\epsilon)\cdot R$}
    % \State $U\gets \{u | (p_j,u)\in E\ s.t. \exists\ p_j\in \mathsf{ALG}\}$
    \STATE  $\overline{C}\gets \overline{C}/2$
    \FOR{$i=1$ to $c\cdot k\log_{\alpha}\frac{\Delta}{\epsilon}$}
    \STATE  $U\gets \bigcup\limits_{p\in \mathsf{ALG}}(N_{out}(p)\cup p)$ and sort $U$ based on their distance from $q$   
    % \STATE $\mathsf{ALG}\gets \varnothing$
    \STATE $\mathsf{ALG}\gets$ the closest $k-1$ points in $\mathsf{ALG}$
    \FOR{each point $u \in U$ in order}
        \IF{$\mathsf{ALG}\bigcup u$ is $(k',\overline{C}/12)$-diverse}
            \STATE $\mathsf{ALG}\gets \mathsf{ALG}\cup u$
            \STATE Break
        \ENDIF
        % \IF{$|\mathsf{ALG}|=k$}
            % \STATE Break
        % \ENDIF
    \ENDFOR
    % \If{$ALG$ doesn't change at this step}
    % \State Break
    % \EndIf
    \ENDFOR
\ENDWHILE
\STATE \textbf{Return} $\mathsf{ALG}$
\end{algorithmic}
\end{algorithm}

%\begin{theorem}\label{thm:dual_diverse_ann}
%Given the graph constructed by Algorithm~\ref{alg:indexing}, the search Algorithm~\ref{alg:dual_search} finds a $(1,0.05C)$-diverse NN solution $\mathsf{ALG}$ satisfying $ALG_k\le \left(\frac{\alpha+1}{\alpha-1}+\epsilon\right)\cdot R$ in $\tilde{O}\left((8\alpha)^dk^3\log\frac{\Delta}{\epsilon}\right)$ time, if there exists a $(1,C)$-diverse solution $\mathsf{OPT}$ with $\mathsf{OPT_k}\le R$.
%\end{theorem}

\begin{proof}[Proof of Theorem~\ref{thm:dual_diverse_ann}]
% For the initial solution $\mathsf{ALG}=\{p_1,...,p_k\}$ selected by the greedy algorithm of Gonzales, we know there doesn't exist a set of $k$ points with minimum pairwise distance greater than $2\min\limits_{p_i,p_j\in \mathsf{ALG}}\rho(p_i,p_j)$. Therefore, for the initialization $\overline{C}=4\min\limits_{p_i,p_j\in \mathsf{ALG}}\rho(p_i,p_j)$, we have $\overline{C}/2\ge C$ where there exists a $(1,C)$-diverse solution $\mathsf{OPT}$ with $\mathsf{OPT_k}\le R$.

After applying the binary search to the initialization algorithm in Lemma~\ref{lm:initialization}, we get an initial $(k',C)$-diverse solution and we know there doesn't exist a $(k',4C)$-diverse solution. Therefore, we set $\overline{C}=4C$ to be the upper bound on the maximal diversity we can achieve.

Then our Algorithm~\ref{alg:dual_search} is basically adding a binary search to Algorithm~\ref{alg:general_search}. Invoking the analysis from Theorem~\ref{thm:diverse_ann}, if there exists a $(k',C)$-diverse solution $\mathsf{OPT}=\{p^*_1,...,p^*_k\}$ with $\mathsf{OPT_k}\le R$, we can find a $(k',C/12)$-diverse solution $\mathsf{ALG}=\{p_1,...,p_k\}$ with $ALG_k\le \left(\frac{\alpha+1}{\alpha-1}+\epsilon\right)\cdot R$ in $O(k\log_{\alpha}\frac{\Delta}{\epsilon})$ steps where each step takes $\tilde{O}((k^3/k')(8\alpha)^d\log\Delta)$ time. As a result, each time when the algorithm enters the while loop on Line 5 in Algorithm~\ref{alg:dual_search}, we know that there doesn't exist a $(k',\overline{C})$-diverse solution with maximal distance smaller than $R$. When we exit the while loop, the current $\overline{C}$ value is at least $1/2$ of the optimal $C$ value, and the current $\mathsf{ALG}$ solution we get is at least $(k',C/24)$-diverse.
\end{proof}

%% file: heuristic.tex
\newcommand{\mcL}{\mathcal{L}}
\newcommand{\mcV}{\mathcal{V}}

\section{Algorithm Implementation} \label{sec:impl}
To conduct our experiments, we provide the heuristic algorithm that we designed for the $k'$-colorful nearest neighbor problem, based on the provable algorithms provided in the main paper. The provable indexing algorithm (\ref{alg:indexing}) has a runtime which is quadratic in the size of the data set and is slow in practice. This situation mimics the original DiskANN algorithm~\cite{jayaram2019diskann}, where the ``slow preprocessing" algorithm  has provable guarantees~\cite{indykxu2024worst} but quadratic running time, and was replaced by a heuristic ``fast preprocessing'' algorithm used in the actual implementation~\cite{DiskANN}. 
%As the aforementioned algorithms are designed to tackle the most general problem and =. Here 
Here, Algorithm~\ref{alg:colorindex} offers a fast method tailored for the $k'$-colorful case, using several heuristics to improve the runtime. In the following section, we present the pseudocode for the procedures: search, index build, and the pruning procedure required for the index build.

%We compare the query performance and accuracy of our algorithm with several baselines, on both real-world shopping data sets and  other data sets where the seller information is synthetically generated. In the remainder of the section, we provide details on our fast implementation of our algorithm, details regarding the datasets and baselines and the results.

%Graph-based similarity search indices are constructed so that greedy search quickly converges to the nearest neighbors  of a query vector $q$. We first describe a natural adaptation of greedy search to account for $k'$-color NN problem and an index construction procedure that allows search to converge to the right answer with relatively few distance comparisons.

\paragraph{Diverse Search.}
Our diverse search procedure, is a greedy graph-based local search method. In our search method, in each step, we maintain a list of best and diverse nodes, ensuring that at most $k'$ points are selected in the list per color. In each iteration of our search algorithm, we choose the best unexplored node and examine its out neighbors. From the union of our current list and the out neighbors, we select the best diverse set of nodes while satisfying the \(k'\)-colorful diversity constraint—meaning no color can have more than \(k'\) points in the updated list. To identify the optimal diverse set from the union, we use a priority queue designed to accommodate the diversity constraint. Below, we present the pseudocode for this diverse priority queue.
\begin{algorithm}
\caption{Insert $(p,d,c)$ into \diversequeue(Q, $L$, $k'$)}
\label{alg:diverse_queue}
\begin{algorithmic}[1]
\STATE \textbf{Input}: Current queue $Q$, tuple $(p,d,c)$ of (point, distance, color) for new insertion, maximum size $L$ of the queue, maximum size $k'$ per color. 
\STATE \textbf{Output}: Updated queue $Q$ after inserting $(p,d,c)$ which maintains the best set of at most $L$ points and at most $k'$ points of each color.
\STATE Let ${\sf count}({{c}}) \gets$ number of elements in $Q$ with the  color $c$.
\STATE Let ${\sf maxDist}({{c}}) \gets$ maximum distance of an element in $Q$ with color $c$.
\IF{${\sf count}({{c}}) < k'$ \textbf{or} $d < {\sf maxDist}({{c}})$}
    \STATE Insert $(p, d, c)$ into $Q$
    \IF{${\sf count}({{c}}) > k'$}
        \STATE Remove the element with the maximum distance in $Q$ having color $c$.
    \ENDIF
\ENDIF
\IF{$|Q| > L$}
    \STATE Remove the element with the maximum distance in $Q$.
\ENDIF
\end{algorithmic}
\end{algorithm}

Building on the previous explanation of the diverse priority queue, we outline the description of our diverse search procedure as follows.

\begin{algorithm}
\caption{$\diversesearch(G,s,q,k',k,L)$}
\label{alg:color_search}
\begin{algorithmic}[1]
\STATE \textbf{Input}: A directed graph $G$, start node $s$, query $q$, max per color parameter $k'$, search list size $L$.
\STATE \textbf{Output}: A set of $k$ points such that there are at most $k'$ points from any color.
\STATE {Initialize} \diversequeue $\mathcal{L} \leftarrow \{(s,D(s,q),col[s]) \}$ with color parameter $k'$ and size parameter $L$.  
\STATE {Initialize} a set of expanded nodes $\mathcal{V} \leftarrow \emptyset$\;
\WHILE{$\mathcal{L} \setminus \mathcal{V} \neq \emptyset$}
        \STATE $\text{Let } p^{*} \leftarrow \argmin\limits_{p \in \mathcal{L} \setminus \mathcal{V}} D(p,q)$\;
        \STATE $\mathcal{V} \leftarrow \mathcal{V} \cup \{p^{*}\}$\;
        % \State $\mathcal{L} \leftarrow \mathcal{L} \cup \{ \left(p, D(p,q) , col[p] \right) \, : \, p \in N_{\text{out}}(p^*)\}$\;
        \STATE Insert $\{\left(p, D(p,q), col[p]\right):p \in N_{\text{out}}(p^*)\}$ to $\mathcal{L}$\;
        % \If{$F_q \cap F_{p^{*}} \neq \emptyset$}{
        %     \State \hspace{0.02cm} $\mathcal{L} \leftarrow \mathcal{L} \cup N_{\text{out}}(p^{*})$\;
        % }
%        \State Pick at most $L$ best diverse candidate out of $\mathcal{L}$. \knote{Add more description here}
\ENDWHILE
\STATE \textbf{Return} $[\text{top } k \text{ NNs from } \mathcal{L}; \mathcal{V}]$\;
\end{algorithmic}
\end{algorithm}

\paragraph{Diverse Prune.} A key subroutine in our index-building algorithm is the prune procedure. Given a node $p$ and a set of potential outgoing edges $\mathcal{V}$, the standard prune procedure removes an edge to a vertex \(w\) if there exists a vertex \(u\) such that an edge \(p \rightarrow u\) exists and the condition \(D(u, w) \leq \frac{D(p, w)}{\alpha}\) is satisfied. Intuitively, this means that to reach \(w\), we would first reach \(u\), thus making multiplicative progress and eliminating the need for the edge \(p \rightarrow w\), which contributes to the sparsity of the graph.

However, to account for diversity, the outgoing edges from the node must also be diverse and enable access to multiple colors. To address this requirement, we modify the standard prune procedure to incorporate the diversity constraint. The details of our revised algorithm are provided next.

\begin{algorithm}[h]
\caption{$\diverseprune(p, \mcV,\alpha,R,m)$}
\label{alg:colorprune}
\begin{algorithmic}[1]
\small
\STATE \textbf{Input}: A point $p$, set $\mcV$, prune parameter $\alpha$, degree parameter $R$, and diversity parameter $m$.
\STATE \textbf{Output}: A subset $\mcV' \subseteq \mcV$ of cardinality at most $R$ to which edges are added.
\STATE Sort all points $u\in \mcV$ based on their distances from $p$ and add them to list $\mcL$ in that order.
\STATE Initialize sets $\blockers[u] \gets \emptyset$ for each $u \in \mcV$.
    \WHILE{$\mcL$ is not empty}
        \STATE $u\gets \argmin\limits_{u\in \mcL}D(u,p)$
        \STATE $\mcV' \gets \mcV' \cup \{u\}$ and $\mcL \gets \mcL \setminus \{u\}$
        \IF{$|\mcV'| = R$}
             \STATE {\bf break}
        \ENDIF
        \FOR{each point $w\in \mcL$}
            \IF{$D(u,w)\le D(p,w)/\alpha$}
                \STATE $\blockers[w]\gets \blockers[w]\cup \{{ col}(u)\}$
                \IF{$|\blockers[w]| = m$ or ${col}(u) = {col}(w)$}
                    \STATE {$\mcL \gets \mcL \setminus \{w\}$}
                \ENDIF
            \ENDIF
        \ENDFOR
    \ENDWHILE
\STATE{\bf Return } $\mcV'$
\end{algorithmic}
\end{algorithm}

\paragraph{Diverse Index.} Our indexing algorithm follows the same approach as the DiskANN ``fast preprocessing" heuristic implementation~\cite{DiskANN}, but we replace the search and prune procedures in their implementation with our diverse search and diverse prune procedures. The details of our index-building procedure are provided below.

\begin{algorithm}[h]
\caption{$\diverseindex(P,\alpha,L,R,m)$}
\label{alg:colorindex}
\begin{algorithmic}[1]
\small
\STATE \textbf{Input}: A set of $n$ points $P=\{p_1,\dots,p_{n} \}$, prune parameter $\alpha$, search list size $L$, degree parameter $R$, and  diversity parameter $m$.
\STATE \textbf{Output}: A directed graph $G$ over $P$ with out-degree at most $ R$.
    \STATE Let $s$ denote the estimated medoid of $P$.\;
    \STATE Initialize $G$ with start node $s$.\;
     
    \FOR{each $p_i\in P$}
        \STATE $\text{Let } [\mathcal{L}; \mathcal{V}] \leftarrow \diversesearch\left(G,s, p_{i}, k' = L/m, L, L \right)$ \; 
        \STATE Let $\mcV' = \diverseprune\left(p_i, \mathcal{V}, \alpha, R, m\right)$.
        \STATE Add node $p_i$ to $G$ and set $N_{\text{out}}(p_{i})= \mcV'$ (out-going edges from $p_i$ to $\mcV'$).\;
        \FOR{$p \in N_{out}(p_i)$}
            \STATE $\text{Update } N_{\text{out}}(p) \leftarrow N_{\text{out}}(p) \cup \{p_i\}$.\;
            \IF{$|N_{\text{out}}(p)| > R$}
                \STATE \textbf{Run} \diverseprune($p, N_{\text{out}}(p), \alpha, R,m$) to update out-neighbors of  $p$.\;
                \ENDIF
\ENDFOR
\ENDFOR
\end{algorithmic}
\end{algorithm}